\documentclass[12pt]{article}
\usepackage{natbib}
\usepackage{amsmath,amssymb,amsfonts}
\usepackage{bm}
\usepackage{booktabs}
\usepackage{fullpage}
\usepackage{algorithm}
\usepackage{algpseudocode}
\usepackage[colorlinks=true,linkcolor=blue,citecolor=blue,urlcolor=blue]{hyperref}
\usepackage{bbm}
\usepackage{graphicx}
\usepackage{xcolor}
\usepackage{url}
\usepackage{amsthm}
\usepackage{multirow}  
\usepackage{longtable} 
\usepackage{caption}   
\usepackage{comment}
\usepackage{subcaption}

\usepackage{threeparttable}
\usepackage{siunitx}
\usepackage{tikz}
\usetikzlibrary{arrows, positioning}
\usetikzlibrary{shapes.geometric}
\usetikzlibrary{positioning,fit,arrows.meta,shapes.geometric}

\newtheorem{remark}{Remark}
\newtheorem{lemma}{Lemma}
\newtheorem{assumption}{Assumption}
\newtheorem{theorem}{Theorem}
\newcommand{\R}{\mathbb{R}}

\newcommand{\tr}{\mathrm{tr}}
\newcommand{\argmin}{\mathop{\mathrm{arg\,min}}}

\title{Sparse $K$-spatial-median clustering for high-dimensional data}

\author{Ping Zhao, Dan Zhuang and
Long Feng\\
Tiangong University, Fujian Normal University, Nankai University
}

\begin{document}
\maketitle

\begin{abstract}
We propose a robust clustering framework for high-dimensional data with heavy tails and a large fraction of irrelevant variables.
The method replaces the mean updates of Lloyd's $K$-means with \emph{spatial medians} to enhance robustness.
For the assignment step, it admits either a Euclidean rule for computational simplicity or a robust Mahalanobis-type metric constructed from the spatial sign covariance matrix to account for heterogeneous scales and feature dependence.
To handle the $p \gg n$ regime, we further introduce a simple \emph{hard feature-exclusion} mechanism that removes weakly separating dimensions based on across-center dispersion, with the exclusion threshold selected automatically via a permutation-based Gap criterion.
Simulation studies under correlated Gaussian and multivariate $t$ models demonstrate that the proposed approach provides competitive clustering accuracy and improved stability relative to $K$-means and sparse $K$-means baselines.



\end{abstract}


\section{Introduction}
High-dimensional clustering arises routinely in modern applications, where the number of variables is often comparable to or substantially larger than the sample size, and only a small subset of features may truly carry cluster-discriminating information \citep{he2021clustermap,park2024vivo,DANESHFAR2024121780}. In such settings, classical clustering methods may deteriorate substantially because irrelevant variables can dominate dissimilarities, while heavy tails, outliers, and contamination can further destabilize mean-based procedures \citep{BenjaminYang2025Survey, de2016survey}.
These challenges have motivated sparse clustering methods that integrate clustering with feature selection or variable screening.
However, in many applications, the central challenge is not sparsity alone, but the need to handle irrelevant features and distributional instability within a unified framework.

Among partition-based methods, $K$-means \citep{Forgy1965} remains one of the most widely used clustering tools because of its simplicity and computational efficiency \citep{MacQueen1967,HartiganWong1979}. Nevertheless, it performs poorly when high-dimensional data contain many irrelevant features or substantial contamination \citep{SunWangFang2012,YANG2024120504}. This weakness is not merely computational: when dissimilarities are aggregated over many coordinates, even weak noise can overwhelm the true low-dimensional separation and obscure the underlying cluster structure. To address this difficulty, sparse clustering methods incorporate feature selection or variable screening directly into the clustering procedure. Existing approaches include model-based penalization \citep{FraleyRaftery2002,PanShen2007,BouveyronBrunetSaumard2014}, subspace and feature-subset methods \citep{ParsonsHaqueLiu2004,ElhamifarVidal2013,WangXu2016NoisySSC}, and objective-function-based extensions of classical clustering algorithms \citep{witten2010sparse,Wild2004FeatureSI}. Among these, sparse $K$-means and its variants are particularly influential, using feature weights together with $\ell_1$-type regularization to suppress irrelevant variables \citep{witten2010sparse,TibshiraniWaltherHastie2001,khan2024determining}. However, these developments remain largely tied to mean-based or weighted $K$-means-type formulations.

A complementary route is to improve robustness by modifying the loss and/or the notion of cluster center. Existing robustifications of $K$-means, including trimmed and weighted variants \citep{CuestaAlbertosGordalizaMatran1997,KondoSalibianBarreraZamar2012,klochkov2021robust,BrodinovaEtAl2019}, reduce sensitivity to outliers and contamination, but they remain closely connected to mean-based or Euclidean-type prototypes. One classical alternative is $K$-median, which replaces the squared $\ell_2$ loss with an $\ell_1$-type distance and uses coordinatewise medians as prototypes \citep{BradleyMangasarianStreet1997,JainDubes1988}. This improves resistance to heavy-tailed coordinates and gross contamination, but it remains inherently axis-aligned: coordinatewise medians and the $\ell_1$ geometry may be sensitive to rotations and correlation structure, which is undesirable when clusters are approximately elliptically distributed. To better respect multivariate geometry while retaining robustness, a natural alternative is the \emph{spatial median} (geometric median), defined as the minimizer of the sum of Euclidean norms. Compared with coordinatewise medians, the spatial median is rotation equivariant, compatible with multivariate elliptical structure, and computationally tractable via Weiszfeld-type iterations \citep{VardiZhang2000,Weiszfeld2009}. More broadly, spatial-sign and spatial-rank methods provide a principled robust framework under elliptical models, in which the spatial median plays the role of a multivariate robust center \citep{Oja2010,feng2026highDimElliptical}.

While robust centers address distributional fragility, they do not by themselves resolve the central difficulty of high-dimensional clustering, since most variables may still be irrelevant. In particular, continuous shrinkage methods can attenuate noisy variables, but may leave many weak coordinates with small nonzero influence. For high-dimensional clustering, it is often more desirable to obtain an explicit retained/excluded feature set, both for interpretability and for sharper separation between signal and noise. This motivates a hard-thresholding approach that performs exact feature exclusion rather than continuous weighting.

In this paper, we propose \textbf{Sparse $K$-spatial-median clustering}, a partition-based clustering framework for high-dimensional data with heavy tails, approximate elliptical geometry, and a large proportion of irrelevant variables. The proposed method combines spatial-median-based cluster centers with an explicit hard-thresholding feature exclusion rule within a single iterative procedure. Specifically, cluster centers are updated as spatial medians, while a simple center-based separation score is computed for each coordinate at every iteration, and variables with scores below a threshold are excluded from the assignment step. In contrast to $\ell_1$-based sparse clustering, this design yields exact feature inclusion and exclusion through a transparent screening rule rather than continuous shrinkage. In this sense, the proposed method is closer to an $\ell_0$-type sparse clustering mechanism, while remaining computationally simple in iterative implementation.

A key practical issue is the choice of the sparsity threshold that governs feature exclusion. To address this, we adopt a permutation-based tuning strategy in the spirit of the Gap statistic, following its common use in sparse clustering \citep{witten2010sparse,BenjaminYang2025Survey} and building on the original framework of \citet{TibshiraniWaltherHastie2001}. Specifically, we evaluate a between-cluster separation functional over candidate thresholds and compare it with a reference distribution generated by column-wise permutations, which preserve marginal feature distributions while destroying multivariate cluster structure. The selected threshold maximizes the resulting Gap criterion, thereby providing a data-driven rule for determining which coordinates are treated as noise.

The main contributions of this paper are as follows:
\begin{itemize}
    \item We develop a robust partition-based clustering framework based on spatial medians, which is better suited to heavy-tailed data and approximately elliptical cluster geometry than mean-based or coordinatewise-median alternatives.

    \item We introduce an explicit hard-thresholding feature exclusion mechanism, which yields exact and interpretable retained/excluded variable sets rather than continuous weight shrinkage.

    \item We propose a permutation-based Gap procedure for automatic threshold selection, providing a data-driven rule for determining which coordinates are treated as noise.

    \item We establish consistency properties of the proposed method and demonstrate its empirical effectiveness through simulation studies and real-data analyses.
\end{itemize}

The remainder of the paper is organized as follows. Section~\ref{sec:method} introduces the proposed Sparse $K$-spatial-median clustering framework. Section~\ref{sec:theory} establishes its corresponding theoretical properties. Section~\ref{sec:experiments} reports simulation results. Section~\ref{sec:realdata} demonstrates the practical performance of the proposed method through analyses of the mice protein dataset and several benchmark datasets. Section~\ref{sec:conclusion} concludes the paper. Appendix~\ref{app:proofs} contains the proofs of the consistency results.

\paragraph*{Notation.}
For a positive integer $q$, write $[q]=\{1,\ldots,q\}$. Vectors are denoted by bold italic symbols, for example $\bm{x}$, $\bm{m}$, and $\bm{\mu}$, whereas matrices are denoted by bold upright symbols, for example $\mathbf{X}$, $\mathbf{A}$, and $\mathbf{\Sigma}$. For $\bm{v}\in\mathbb R^q$, $v_j$ is its $j$th coordinate and $\bm{v}_S=(v_j)_{j\in S}$ is its subvector on $S\subset[q]$; $S^c=[q]\setminus S$ and $|S|$ is the cardinality of $S$. Unless otherwise stated, $\|\bm{v}\|=\|\bm{v}\|_2$ is the Euclidean norm, while $\|\bm{v}\|_1$ and $\|\bm{v}\|_\infty$ denote the usual $\ell_1$ and maximum norms. For a positive semidefinite matrix $\mathbf{A}$, $\|\bm{v}\|_{\mathbf{A}}=(\bm{v}^\top\mathbf{A}\bm{v})^{1/2}$. For a matrix $\mathbf{M}$, $\mathbf{M}^\top$ denotes transpose, $\|\mathbf{M}\|_{\rm op}$ denotes spectral norm, $\tr(\mathbf{M})$ denotes trace, and $\lambda_{\min}(\mathbf{M})$ and $\lambda_{\max}(\mathbf{M})$ denote the smallest and largest eigenvalues when $\mathbf{M}$ is symmetric. We write $\mathbf{M}\succeq0$ for positive semidefiniteness and $\mathbf{A}\preceq\mathbf{B}$ when $\mathbf{B}-\mathbf{A}\succeq0$; $\mathbf{I}_q$ is the $q\times q$ identity matrix.

\section{Methodology} \label{sec:method}

\subsection{Conventional Partitioning Algorithms}
Given observations $\{\bm{x}_i\}_{i=1}^n\subset\mathbb{R}^p$, equivalently the data matrix $\mathbf{X}=(\bm{x}_1,\ldots,\bm{x}_n)^\top\in\mathbb{R}^{n\times p}$, and a target number of clusters $K$, we seek a partition $\{C_k\}_{k=1}^K$ (equivalently assignments $c_i\in\{1,\ldots,K\}$) and corresponding center vectors $\{\bm{m}_k\}_{k=1}^K$. A convenient way to view a large family of partitional clustering methods is through the generic empirical risk
\begin{equation*}
\min_{\{c_i\},\,\{\bm{m}_k\}}\;\sum_{i=1}^n \rho\!\left(\|\bm{x}_i-\bm{m}_{c_i}\|\right),
\end{equation*}
where $\rho(\cdot)$ controls tail-robustness (loss growth), and the norm determines the geometry of the clusters. This lens clarifies that many classical algorithms differ mainly in (i) how strongly they penalize large residuals and (ii) whether they respect the correlation/rotation structure of the data.

The most widely used special case is $K$-means, obtained by taking the Euclidean norm and the squared loss $\rho(t)=t^2$. The resulting objective is
\[
\min_{\{c_i\},\,\{\bm{m}_k\}}\;\sum_{i=1}^n \|\bm{x}_i-\bm{m}_{c_i}\|_2^2,
\]
for which the optimal center update (given assignments) is the arithmetic mean of each cluster. Computationally, the standard Lloyd-type alternating scheme cycles between assigning each point to its nearest center (in squared Euclidean distance) and updating each center by the within-cluster mean. While attractive for its simplicity and speed, $K$-means is well known to be sensitive to heavy-tailed noise and outliers because the squared loss amplifies large deviations; in high-noise coordinates, a small number of extreme observations can dominate both the assignment and the center updates.

A robust alternative is $K$-medians, typically defined by replacing the squared Euclidean criterion with an $\ell_1$ geometry, i.e.
\[
\min_{\{c_i\},\,\{\bm{m}_k\}}\;\sum_{i=1}^n \|\bm{x}_i-\bm{m}_{c_i}\|_1.
\]
Given assignments, each coordinate of $\bm{m}_k$ is updated by the univariate median of the corresponding coordinate within cluster $k$, yielding bounded-influence behavior at the coordinate level and substantially improved stability under heavy tails and gross contamination. However, the $\ell_1$ geometry is not rotation equivariant: when the data are correlated, or when the natural cluster shapes are ellipsoidal after an unknown linear transformation, coordinate-wise medians can be statistically inefficient because they treat each coordinate independently and do not adapt to the orientation of the scatter.

This motivates $K$-\emph{spatial-median} clustering, which keeps the Euclidean geometry but replaces the squared loss by the unsquared Euclidean norm. The center of cluster $k$ is defined as the spatial median (geometric median)
\begin{equation}
\bm{m}_k \;=\; \argmin_{\bm{m} \in \mathbb{R}^p}\;\sum_{i\in C_k}\|\bm{x}_i-\bm{m}\|_2,
\label{eq:spmed}
\end{equation}
leading to the overall objective $\sum_i \|\bm{x}_i-\bm{m}_{c_i}\|_2$. Relative to $K$-means, the linear growth of $\|\bm{x}_i-\bm{m}\|_2$ yields substantially improved robustness to heavy-tailed residuals, while the Euclidean geometry makes the estimator rotation equivariant. This combination is particularly attractive under elliptical symmetry, where the cluster contours are naturally ellipsoidal up to a common scatter structure: unlike coordinate-wise medians, the spatial median respects the global geometry and typically provides more stable centers when correlation is non-negligible.



Algorithmically, $K$-spatial-median clustering follows the same alternating-minimization template as $K$-means. Given the current assignments $\{c_i\}$, each center $\bm{m}_k$ is updated by solving \eqref{eq:spmed}, for example via a Weiszfeld-type iterative scheme. Given the current centers $\{\bm{m}_k\}$, each observation is then reassigned to its nearest center.
Although this basic spatial-median formulation already improves robustness in the center-update step, it does not by itself resolve two difficulties that are central in high-dimensional settings. First, when variables have heterogeneous scales or are strongly correlated, a plain Euclidean assignment rule may not adequately reflect the underlying elliptical geometry of the clusters. Second, when a large proportion of features is irrelevant, treating all coordinates equally can substantially dilute the clustering signal.

To address these two issues, we develop two complementary methods within a common spatial-median framework. The first is \emph{Spatial-Median Clustering with Common Spatial-Sign Covariance (SM--SSCM)}, which improves the assignment step through a common SSCM-based metric. The second is \emph{Sparse Spatial-Median Clustering (Sparse--SM)}, which performs explicit hard-threshold feature exclusion to reduce the influence of irrelevant variables. Figure~\ref{fig:four_module_framework} provides a schematic overview of the proposed methodology. The figure highlights the common spatial-median update shared by both methods, together with the two distinct mechanisms introduced to address the main high-dimensional difficulties: SSCM-based metric adaptation in SM--SSCM and hard-threshold feature exclusion in Sparse--SM\@. It also shows how the threshold parameter and the number of clusters are selected through the outer data-driven modules. We next describe these components in detail.

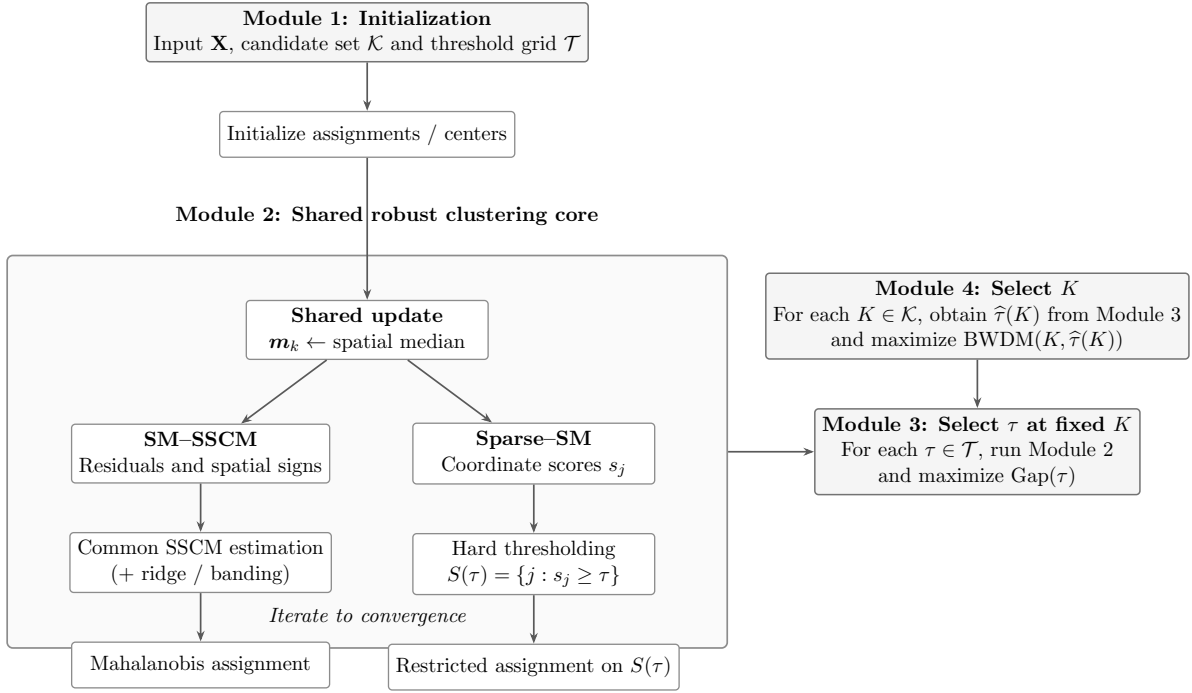
\begin{figure}[t]
\centering
\resizebox{0.96\textwidth}{!}{
\begin{tikzpicture}[
  font=\small,
  node distance=9mm and 12mm,
  box/.style={
    rectangle, rounded corners=2pt, draw=black!55, line width=0.8pt,
    align=center, minimum height=0.95cm, fill=gray!8
  },
  smallbox/.style={
    rectangle, rounded corners=2pt, draw=black!45, line width=0.7pt,
    align=center, minimum height=0.85cm, fill=white
  },
  bigbox/.style={
    rectangle, rounded corners=3pt, draw=black!45, line width=0.9pt,
    inner sep=10pt, fill=gray!4
  },
  arrow/.style={-{Stealth[length=2mm]}, line width=0.9pt, draw=black!65}
]

\node[box, minimum width=6.4cm] (m1) {\textbf{Module 1: Initialization}\\
Input $\mathbf{X}$, candidate set $\mathcal{K}$ and threshold grid $\mathcal{T}$};

\node[smallbox, below=of m1, minimum width=5.4cm] (init) {Initialize assignments / centers};

\node[bigbox, below=18mm of init, minimum width=13.4cm, minimum height=7.3cm] (m2frame) {};

\node[font=\small\bfseries, above=4mm of m2frame.north, xshift=3.7mm] (m2title)
{Module 2: Shared robust clustering core};

\node[smallbox, minimum width=4.5cm] at ([yshift=-14mm]m2frame.north) (center)
{\textbf{Shared update}\\
$\bm{m}_k \leftarrow$ spatial median};

\node[smallbox, minimum width=4.5cm] at ([xshift=-31mm,yshift=-37mm]m2frame.north) (sm1)
{\textbf{SM--SSCM}\\
Residuals and spatial signs};

\node[smallbox, below=of sm1, minimum width=4.5cm] (sm2)
{Common SSCM estimation\\
(+ ridge / banding)};

\node[smallbox, below=of sm2, minimum width=4.8cm] (sm3)
{Mahalanobis assignment};

\node[smallbox, minimum width=4.5cm] at ([xshift=31mm,yshift=-37mm]m2frame.north) (sp1)
{\textbf{Sparse--SM}\\
Coordinate scores $s_j$};

\node[smallbox, below=of sp1, minimum width=4.5cm] (sp2)
{Hard thresholding\\
$S(\tau)=\{j:s_j\ge\tau\}$};

\node[smallbox, below=of sp2, minimum width=4.8cm] (sp3)
{Restricted assignment on $S(\tau)$};

\node[font=\small\itshape] at ([yshift=6mm]m2frame.south)
{Iterate to convergence};

\node[box, right=16mm of m2frame, minimum width=5.3cm] (m3)
{\textbf{Module 3: Select $\tau$ at fixed $K$}\\
For each $\tau\in\mathcal{T}$, run Module~2\\
and maximize Gap$(\tau)$};

\node[box, above=of m3, minimum width=5.3cm] (m4)
{\textbf{Module 4: Select $K$}\\
For each $K\in\mathcal{K}$, obtain $\widehat\tau(K)$ from Module~3\\
and maximize BWDM$(K,\widehat\tau(K))$};

\draw[arrow] (m1) -- (init);
\draw[arrow] (init) -- (center);

\draw[arrow] (center) -- (sm1);
\draw[arrow] (center) -- (sp1);

\draw[arrow] (sm1) -- (sm2);
\draw[arrow] (sm2) -- (sm3);

\draw[arrow] (sp1) -- (sp2);
\draw[arrow] (sp2) -- (sp3);

\draw[arrow] (m2frame.east) -- (m3.west);
\draw[arrow] (m4.south) -- (m3.north);

\end{tikzpicture}
}
\caption{Framework of the proposed methods. Module~2 is the shared spatial-median clustering core. Within this core, SM--SSCM achieves metric robustness through a common SSCM-based assignment rule, whereas Sparse--SM achieves dimension robustness through hard-threshold feature exclusion. Module~3 performs data-driven threshold selection for a fixed number of clusters, and Module~4 selects the number of clusters.}
\label{fig:four_module_framework}
\end{figure}



\subsection{SM--SSCM: Spatial-Median Clustering with Common Spatial-Sign Covariance}

The basic $K$-spatial-median formulation improves robustness in the center-update step, but it still relies on a plain Euclidean assignment rule. In high-dimensional settings, such a rule may be inadequate when variables have heterogeneous scales or exhibit substantial dependence, because Euclidean distances do not adapt to the underlying elliptical geometry of the clusters. To address this issue, we propose SM--SSCM, which replaces the Euclidean assignment step by a common Mahalanobis-type rule constructed from a spatial-sign covariance estimate shared across clusters, while retaining spatial medians as robust cluster centers.

With the current centers \(\{\bm{m}_k\}_{k=1}^K\) fixed, each observation is assigned to the cluster minimizing the corresponding SSCM-based distance. The centers are then updated as the spatial medians of their assigned observations. These two steps are alternated until convergence.

To connect the assignment rule to robust residual geometry, define the residuals under the current partition by
\[
\bm{r}_i=\bm{x}_i-\bm{m}_{c_i},\qquad i=1,\ldots,n,
\]
and the spatial-sign transform by
\[
\bm{u}_i=\begin{cases}
\bm{r}_i/\|\bm{r}_i\|_2, & \bm{r}_i\neq 0,\\
0, & \bm{r}_i=0.
\end{cases}
\]
We estimate a common scatter surrogate by the regularized spatial-sign covariance matrix
\begin{equation}
\widehat{\mathbf{\Sigma}}
=\frac{1}{n}\sum_{i=1}^n \bm{u}_i \bm{u}_i^\top + \lambda \mathbf{I}_p,
\label{eq:sscm}
\end{equation}
where \(\lambda>0\) is a ridge parameter, and one may optionally apply a banding or thresholding operator to \(\widehat{\mathbf{\Sigma}}\) in high dimensions.

For the assignment step, we consider the quadratic rule
\begin{equation}
d_{ik}=(\bm{x}_i-\bm{m}_k)^\top \mathbf{A}(\bm{x}_i-\bm{m}_k),
\qquad
c_i \leftarrow \arg\min_{k} d_{ik},
\label{eq:dist_general}
\end{equation}
where \(\mathbf{A}\) is a positive definite weighting matrix. When \(\mathbf{A}=\mathbf{I}_p\), \eqref{eq:dist_general} reduces to the standard Euclidean rule
\[
d_{ik}=\|\bm{x}_i-\bm{m}_k\|_2^2,
\]
which preserves the simplicity and scalability of Lloyd-type updates. Alternatively, taking \(\mathbf{A}=\widehat{\mathbf{\Sigma}}^{-1}\) yields the Mahalanobis distance
\begin{equation}
d_{ik}=(\bm{x}_i-\bm{m}_k)^\top \widehat{\mathbf{\Sigma}}^{-1}(\bm{x}_i-\bm{m}_k),
\label{eq:mahdist}
\end{equation}
which accounts for scale heterogeneity and dependence among features. Equivalently, this metric may be interpreted as performing clustering in a whitened feature space induced by \(\widehat{\mathbf{\Sigma}}^{-1/2}\).

To enhance numerical stability, we adopt a \emph{safe} update strategy. In particular, if a cluster becomes empty, or if the assigned points collapse (for example, exhibiting zero dispersion across all coordinates), the update is replaced by a representative data point, such as the observation closest to the previous center. Although the objective in \eqref{eq:spmed} remains well defined in such degenerate cases, standard Weiszfeld-type iterations may become ill-posed because of vanishing denominators. This safeguard ensures that the SSCM-based distances remain well defined and that the iterative procedure remains numerically stable. The resulting alternating procedure for SM--SSCM is summarized in Algorithm~\ref{alg:robust}.

\begin{remark}[Geometric interpretation of the center update]
For a fixed partition, the spatial median update defined in \eqref{eq:spmed} is characterized by the subgradient balance
\[
0 \in \sum_{i:c_i=k} \frac{\bm{x}_i-\bm{m}_k}{\|\bm{x}_i-\bm{m}_k\|_2},
\]
which shows that the center is determined by the directional balance of residuals. Consequently, observations with large residual norms affect the update primarily through their directions rather than their magnitudes, which explains the robustness of the spatial-median center update. In practice, \eqref{eq:spmed} is computed numerically by a Weiszfeld-type iteration.
\end{remark}

\begin{algorithm}[t]
\caption{Spatial-Median Clustering with Common SSCM Metric (SM--SSCM)}
\label{alg:robust}
\begin{algorithmic}[1]
\Require Data $\mathbf{X}\in\R^{n\times p}$, number of clusters $K$, maximum number of iterations $T$, ridge parameter $\lambda$
\State Initialize assignments $c_i \sim \mathrm{Unif}\{1,\ldots,K\}$
\For{$t=1,\ldots,T$}
  \State \textbf{Center update:} for each $k$, compute the spatial median $\bm{m}_k$ from $\{\bm{x}_i: c_i=k\}$ using \eqref{eq:spmed} (with safe fallback if needed).
  \State \textbf{SSCM update:} compute residuals $\bm{r}_i=\bm{x}_i-\bm{m}_{c_i}$, spatial signs $\bm{u}_i$, and \(\widehat{\mathbf{\Sigma}}\) by \eqref{eq:sscm}.
  \State \textbf{Assignment:} compute $d_{ik}$ via \eqref{eq:mahdist} and set $c_i \leftarrow \arg\min_k d_{ik}$.
  \State \textbf{Empty-cluster fix:} if any cluster is empty, reassign a representative point to that cluster.
  \If{assignments are unchanged}
    \State \textbf{break}
  \EndIf
\EndFor
\State \Return $\{c_i\}$, $\{\bm{m}_k\}$
\end{algorithmic}
\end{algorithm}

\subsection{Sparse--SM: Sparse Spatial-Median Clustering}

The SM--SSCM procedure addresses robustness under heterogeneous scales and feature dependence by modifying the assignment metric, but it still treats all coordinates equally. In high-dimensional settings, this can be problematic when a large proportion of variables is irrelevant, because even a robust distance may deteriorate once the cumulative contribution of noise coordinates overwhelms the weak separation signal. This motivates a complementary extension that incorporates explicit feature exclusion into the spatial-median framework.

Sparse clustering is a natural response to this difficulty. A canonical example is sparse $K$-means, which introduces feature weights and enforces sparsity through an $\ell_1$-type penalty or constraint so that only a subset of coordinates contributes substantially to the clustering criterion. In practice, this leads to an alternating procedure that updates clusters given the current weights and then updates the weights given the current clusters, typically through a soft-thresholding-type operation. While effective and widely used, such approaches produce \emph{soft} sparsity: many coordinates are shrunk toward zero, but the distinction between selected and excluded variables may remain less transparent, and the tuning parameter may be difficult to calibrate in robust settings.

Related ideas also appear in sparse $K$-median methodology. In particular, \citet{Wild2004FeatureSI} incorporated feature selection directly into the $k$-median framework by embedding sparsity into the clustering optimization problem. That line of work provides a principled way to combine robustness and feature selection, but it also inherits the typical behavior of \(\ell_1\)-type regularization, including shrinkage bias and sensitivity of the selected set to the penalty level.

In contrast, we propose a deliberately \emph{hard-thresholding} exclusion strategy that yields an explicit and easily interpretable active set at each iteration, avoiding \(\ell_1\) shrinkage altogether. The key quantity is a simple separation score computed from the current centers. For each coordinate \(j\in\{1,\ldots,p\}\), define
\begin{equation}
s_j \;=\; \sum_{k=1}^K \big|m_{kj}-\bar m_j\big|,
\qquad
\bar m_j=\frac{1}{K}\sum_{k=1}^K m_{kj}.
\label{eq:sep}
\end{equation}
Intuitively, \(s_j\) measures the extent to which the cluster centers separate along coordinate \(j\). If all centers are nearly equal in that coordinate, then \(s_j\) is small and the coordinate is plausibly non-informative for clustering. The absolute-deviation form in \eqref{eq:sep} aligns naturally with the spatial-median geometry and avoids undue amplification of extreme center differences.

For a threshold \(\tau>0\), define the active set
\[
S(\tau)=\{j:\ s_j\ge \tau\},
\]
and exclude coordinates in \(S(\tau)^c\) from the assignment step. Specifically, the assignment update becomes
\begin{equation}
c_i \leftarrow \arg\min_{k}\ \|(\bm{x}_i)_{S(\tau)}-(\bm{m}_k)_{S(\tau)}\|_2^2.
\label{eq:sparse_assignment}
\end{equation}
The center updates retain the spatial-median form and are computed in the full feature space. Thus, robustness continues to enter through the spatial-median centers, whereas sparsity enters through hard dimension exclusion in the assignment distances.

If desired, sparsity can also be made explicit at the parameter level by resetting the excluded coordinates to a common baseline, for example the corresponding coordinate of the overall spatial median, so that all clusters share the same value on \(S(\tau)^c\). This additional step is not essential for the assignment rule itself, but it makes the exclusion mechanism fully transparent in the center representation.

Putting the pieces together, Sparse--SM alternates among three operations: (i) updating the cluster centers by spatial medians, with the same safe fallback used in SM--SSCM for empty or degenerate clusters; (ii) computing the separation scores \(\{s_j\}\) from the current centers and forming the active set \(S(\tau)\); and (iii) updating assignments using only the active coordinates. This yields a robust-and-sparse analogue of Lloyd's algorithm: robustness is inherited from the spatial-median center updates, while sparsity is introduced through explicit hard exclusion rather than through continuously shrunk feature weights.

To select the threshold \(\tau\) in a data-driven manner, we employ a permutation-based Gap criterion; see Section~\ref{subsec:tau_selection} for details. For each candidate \(\tau\) on a pre-specified grid, we run the Sparse--SM iterations to convergence and compute the observed clustering criterion \(O(\tau)\). We then independently permute each column of the data matrix, rerun the same procedure on the permuted data, and obtain the reference values \(O^{(b)}(\tau)\), \(b=1,\ldots,B\). The selected threshold \(\widehat{\tau}\) maximizes the resulting Gap score. Because this strategy preserves the marginal distribution of each coordinate while disrupting the joint clustering structure, it is well suited to high-dimensional settings with potentially heavy-tailed features. The resulting sparse dimension-exclusion procedure, together with permutation-based tuning of \(\tau\), is summarized in Algorithm~\ref{alg:sparse}.


\begin{algorithm}[t]
\caption{Sparse Spatial-Median Clustering by Dimension Exclusion (Sparse--SM) with Gap Tuning}
\label{alg:sparse}
\begin{algorithmic}[1]
\Require Data $\mathbf{X}\in\R^{n\times p}$, number of clusters $K$, threshold grid $\{\tau_\ell\}$, number of permutations $B$
\For{each $\tau_\ell$}
  \State Run Sparse--SM iterations:
  \Statex \quad (i) update spatial-median centers $\{\bm{m}_k\}$;
  \Statex \quad (ii) compute separation scores $s_j$ via \eqref{eq:sep};
  \Statex \quad (iii) form \(S(\tau_\ell)\), exclude coordinates with \(s_j<\tau_\ell\), and reassign by Euclidean distance on the remaining dimensions as in \eqref{eq:sparse_assignment}.
  \State Compute the observed criterion \(O(\tau_\ell)\) via \eqref{eq:Otau}.
  \For{$b=1,\ldots,B$}
    \State Form permuted data \(\mathbf{X}^{(b)}\) by independently permuting each column of \(\mathbf{X}\).
    \State Run Sparse--SM on \(\mathbf{X}^{(b)}\) at \(\tau_\ell\) and compute \(O^{(b)}(\tau_\ell)\).
  \EndFor
  \State Compute \(\mathrm{Gap}(\tau_\ell)\) via \eqref{eq:gap}.
\EndFor
\State Select \(\widehat{\tau}=\arg\max_{\tau_\ell}\mathrm{Gap}(\tau_\ell)\) and rerun Sparse--SM at \(\widehat{\tau}\).
\State \Return assignments, centers, excluded dimensions, and \(\widehat{\tau}\).
\end{algorithmic}
\end{algorithm}


\subsection{Selection of thresholding parameter \texorpdfstring{$\tau$}{tau}}
\label{subsec:tau_selection}

Finally, we require a data-driven rule to select the threshold \(\tau\). We adapt the Gap principle to this tuning problem. Let \(\widehat{\bm{m}}_k(\tau)\) denote the final centers returned by the sparse algorithm at threshold \(\tau\), and let \(\widehat{\bm{m}}(\tau)\) be the overall spatial median, computed under the same coordinate convention as the algorithm. Define the between-cluster separation functional
\begin{equation}
O(\tau) \;=\; \sum_{k=1}^K n_k \big\|\widehat{\bm{m}}_k(\tau)-\widehat{\bm{m}}(\tau)\big\|_2^2,
\label{eq:Otau}
\end{equation}
where \(n_k\) is the size of cluster \(k\). Larger values of \(O(\tau)\) indicate stronger separation among the fitted cluster centers. Note that \(\widehat{\bm{m}}(\tau)\) is the spatial median of the full sample rather than the weighted average of the fitted cluster centers, reflecting the robust geometry underlying the procedure.

To calibrate the level of separation expected under no genuine cluster structure, we generate reference data sets by independently permuting each coordinate (column-wise permutation), thereby preserving the marginal distribution of each feature while destroying the joint clustering structure across observations. For each permutation replicate \(b=1,\ldots,B\), we run the same sparse procedure at the same \(\tau\) and obtain \(O^{(b)}(\tau)\). We then define
\begin{equation}
\mathrm{Gap}(\tau)
\;=\;
\log O(\tau)\;-\;\frac{1}{B}\sum_{b=1}^B \log O^{(b)}(\tau),
\label{eq:gap}
\end{equation}
and select \(\tau\) by maximizing \(\mathrm{Gap}(\tau)\) over a pre-specified candidate grid.

Unlike the classical Gap statistic, which is based on within-cluster dispersion, we employ a between-cluster separation functional because the sparsity parameter \(\tau\) directly controls the coordinates contributing to cluster discrimination. Since the null reference preserves the marginal distribution of each feature while disrupting the joint clustering structure, this criterion is aligned with the goal of identifying coordinates that genuinely contribute to cluster separation rather than reflecting spurious alignment in the sample. Operationally, the selected threshold yields an interpretable model: the retained set \(S(\tau)\) can be reported as the estimated clustering-relevant variables, and excluded dimensions are treated as pure noise in the clustering geometry. Although the resulting procedure does not correspond to the exact minimization of a single global objective, it may be viewed as a block-coordinate heuristic alternating between clustering and feature screening.

\subsection{Selecting the number of clusters \texorpdfstring{$K$}{K}} \label{subsec:K_selection}

To select the number of clusters in sparse \(K\)-spatial-median clustering, we adapt the spatial-median stopping rule of \citet{GabrWillisBaragilly2025}, which constructs a nonparametric index from (i) an average \emph{between-cluster} distance among cluster spatial medians and (ii) an average \emph{within-cluster} distance from observations to their assigned cluster spatial median. For each candidate \(K \in \{2,\ldots,K_{\max}\}\), we first tune the sparsity threshold \(\tau\) using the permutation-based Gap criterion in Section~\ref{subsec:tau_selection}, obtaining a retained coordinate set \(S(\widehat{\tau}(K))\). We then run the sparse \(K\)-spatial-median algorithm restricted to \(S(\widehat{\tau}(K))\), obtaining cluster assignments and cluster spatial medians \(\{\widehat{\bm{m}}_k(K,\widehat{\tau}(K))\}_{k=1}^K\).

For the purpose of selecting \(K\), both the between-cluster and within-cluster distances are evaluated in the retained subspace \(S(\widehat{\tau}(K))\). Let \(\widehat{\bm{m}}_k(K,\tau)\) denote the spatial median of cluster \(k\) computed in that retained subspace, and let \(n_k\) be the resulting cluster size. We define the average between-median distance
\[
\mathrm{ABDM}(K,\tau)=\binom{K}{2}^{-1}\sum_{1\le k<\ell\le K}
\left\|\widehat{\bm{m}}_k(K,\tau)-\widehat{\bm{m}}_\ell(K,\tau)\right\|_2,
\]
and the average within-median distance
\[
\mathrm{AWDM}(K,\tau)=\frac{1}{n}\sum_{k=1}^K\sum_{i:c_i=k}
\left\|\bm{x}_{i,S(\tau)}-\widehat{\bm{m}}_k(K,\tau)\right\|_2.
\]
Following the degree-of-freedom normalization in \citet{GabrWillisBaragilly2025}, we form
\[
\mathrm{BWDM}(K,\tau)
=
\frac{\mathrm{ABDM}(K,\tau)/(K-1)}
{\mathrm{AWDM}(K,\tau)/(n-K)},
\]
and select
\[
\widehat{K}
=
\arg\max_{K\in\{2,\ldots,K_{\max}\}}
\mathrm{BWDM}\bigl(K,\widehat{\tau}(K)\bigr),
\]
where \(\widehat{\tau}(K)\) is the Gap-selected sparsity level at fixed \(K\), so that sparsity is re-tuned for each candidate value of \(K\).

\section{Theoretical Properties}
\label{sec:theory}

This section establishes the main theoretical properties of the proposed methods. We first study the clustering consistency of SM--SSCM with its SSCM-based assignment metric. We then develop feature-selection consistency and clustering consistency results for the hard-threshold sparse $K$-spatial-median procedure (Sparse--SM). Throughout this section, $K$ is fixed, and all asymptotic statements are understood as $n \to \infty$.

Our analysis focuses on the max--min initialized SM--SSCM iteration with the regularized SSCM metric $\mathbf{A}_t=\widehat{\mathbf{\Sigma}}_t^{-1}$, where $\widehat{\mathbf{\Sigma}}_t$ is defined as in \eqref{eq:sscm} with a fixed ridge parameter $\lambda>0$, and on the hard-threshold sparse $K$-spatial-median procedure with a deterministic threshold sequence $\tau_n$. The optional banding/thresholding of $\widehat{\mathbf{\Sigma}}_t$, the random-restart implementation, and the data-driven Gap selector used in the empirical study are not analyzed here. To allow increasing-signal asymptotics, the cluster centers and the derived separation and score quantities are permitted to depend on $n$, although we suppress this dependence in the notation whenever no confusion can arise. All minimizers, nearest-center ties, and nonunique spatial medians are resolved by a fixed deterministic rule.

\subsection{Unified assumptions and basic concentration}
\label{subsec:theory_assumptions}

We begin with a unified probabilistic framework for both SM--SSCM and Sparse--SM.

\begin{assumption}[Mixture elliptical model and basic regularity]
\label{as:theory_A1}
Let $z_i\in\{1,\ldots,K\}$ be i.i.d.\ with $\Pr(z_i=r)=\pi_r$ and $\pi_{\min}:=\min_{1\le r\le K}\pi_r>0$.
Conditional on $z_i=r$, the observation $\bm{X}_i\in\mathbb{R}^p$ admits the representation
\begin{equation}
\bm{X}_i=\bm{\mu}_r+\mathbf{\Sigma}_r^{1/2}(R_i\bm{U}_i),
\label{eq:theory_model}
\end{equation}
where $\bm{\mu}_r\in\mathbb{R}^p$, $\mathbf{\Sigma}_r\succeq 0$, $\bm{U}_i\sim \mathrm{Unif}(\mathbb{S}^{p-1})$, $R_i\ge 0$, and $R_i\perp \bm{U}_i$.
Define
\[
\lambda_{\max}:=\max_{1\le r\le K}\lambda_{\max}(\mathbf{\Sigma}_r)<\infty.
\]
\end{assumption}

\begin{assumption}[Radial tail control]
\label{as:theory_A2}
One of the following holds for the radial variable $R_i$ in \eqref{eq:theory_model}:
\begin{enumerate}
\item[(T$_\alpha$)] (\emph{sub-Weibull tail}) There exist $\alpha>0$ and $\kappa>0$ such that for all $t\ge 0$,
\begin{equation*}
\Pr(R_i>t)\le 2\exp\!\big(-(t/\kappa)^\alpha\big).
\end{equation*}
\item[(T$_\nu$)] (\emph{polynomial tail}) There exist $\nu>0$, $C>0$, and $t_0\ge 0$ such that for all $t\ge t_0$,
\begin{equation*}
\Pr(R_i>t)\le Ct^{-\nu}.
\end{equation*}
\end{enumerate}
\end{assumption}

Fix $\varepsilon>0$ and let $\delta_n:=n^{-\varepsilon}$. Define the deviation radius
\begin{equation*}
a_n=
\begin{cases}
\sqrt{\lambda_{\max}}\,\kappa\big(\log(2n/\delta_n)\big)^{1/\alpha},
& \text{under (T$_\alpha$)},\\[6pt]
\sqrt{\lambda_{\max}}\,\big(Cn/\delta_n\big)^{1/\nu},
& \text{under (T$_\nu$)},
\end{cases}
\end{equation*}
and the high-probability event
\begin{equation*}
\mathcal{E}_n:=\Big\{\max_{1\le i\le n}\|\bm{X}_i-\bm{\mu}_{z_i}\|\le a_n\Big\}.
\end{equation*}

The next lemma shows that $\mathcal{E}_n$ occurs with probability tending to one.

\begin{lemma}[Maximal deviation bound]
\label{lem:theory_En}
Under Assumptions~\ref{as:theory_A1}--\ref{as:theory_A2},
\[
\Pr(\mathcal{E}_n)\ge 1-\delta_n
\]
for all sufficiently large $n$; under (T$_\alpha$), the bound holds for every $n$.
\end{lemma}

Let $C_{r,T}:=\{i:z_i=r\}$ denote the true cluster indexed by $r$, and define
\begin{equation*}
\mathcal{F}_n:=\Big\{\min_{1\le r\le K}|C_{r,T}|\ge 1\Big\}.
\end{equation*}
Since $\pi_{\min}>0$ and $K$ is fixed, we have $\Pr(\mathcal{F}_n)\to 1$. Throughout, an estimated partition $\{\hat C_k\}_{k=1}^K$ is said to be \emph{clustering consistent} if there exists a permutation $\sigma$ of $\{1,\ldots,K\}$ such that
\begin{equation}
\Pr\!\Big(\hat C_k=C_{\sigma(k),T}\ \text{for all }k=1,\ldots,K\Big)\to 1.
\label{eq:theory_clust_consistency}
\end{equation}

\subsection{Clustering consistency of SM--SSCM}
\label{subsec:theory_smsscm}

We now turn to SM--SSCM with the regularized inverse-SSCM assignment metric. At iteration $t$, let
\[
\widehat{\mathbf{S}}_t=\frac{1}{n}\sum_{i=1}^n \bm{u}_i^{(t)}(\bm{u}_i^{(t)})^\top,
\qquad
\widehat{\mathbf{\Sigma}}_t=\widehat{\mathbf{S}}_t+\lambda \mathbf{I}_p,
\qquad
\mathbf{A}_t=\widehat{\mathbf{\Sigma}}_t^{-1},
\]
where $\bm{u}_i^{(t)}$ is the spatial sign of the current residual and $\lambda>0$ is fixed. The SM--SSCM assignment step compares
\[
(\bm{x}_i-\bm{m}_k)^\top \mathbf{A}_t(\bm{x}_i-\bm{m}_k),
\qquad k=1,\ldots,K.
\]
Let
\begin{equation*}
d:=\min_{r\neq s}\|\bm{\mu}_r-\bm{\mu}_s\|
\end{equation*}
denote the minimal Euclidean separation between the true cluster centers.

To make the SSCM information explicit, let $\bm{Y}^{(r)}=\mathbf{\Sigma}_r^{1/2}(R\bm{U})$ denote a centered within-cluster observation from cluster $r$, where $\bm{U}\sim\mathrm{Unif}(\mathbb S^{p-1})$ and $R$ has the corresponding radial distribution. Define
\[
\bm{V}^{(r)}=
\begin{cases}
\bm{Y}^{(r)}/\|\bm{Y}^{(r)}\|, & \bm{Y}^{(r)}\ne0,\\
0, & \bm{Y}^{(r)}=0,
\end{cases}
\]
and define the population common spatial-sign covariance matrix
\begin{equation*}
\mathbf{\Gamma}_{\rm sgn}:=\sum_{r=1}^K\pi_r\,\mathbb E\{\bm{V}^{(r)}\bm{V}^{(r)\top}\}.
\end{equation*}
For the same ridge parameter $\lambda$ as in the algorithm, set
\begin{equation*}
\mathbf{\Sigma}_\star:=\mathbf{\Gamma}_{\rm sgn}+\lambda \mathbf{I}_p,
\qquad
\mathbf{A}_\star:=\mathbf{\Sigma}_\star^{-1},
\qquad
\|\bm{v}\|_{\mathbf{A}_\star}:=(\bm{v}^\top \mathbf{A}_\star \bm{v})^{1/2},
\end{equation*}
and
\begin{equation*}
d_{\rm sgn}:=\min_{r\ne s}\|\bm{\mu}_r-\bm{\mu}_s\|_{\mathbf{A}_\star}.
\end{equation*}
Thus $d_{\rm sgn}$ is the cluster separation measured after whitening by the population spatial-sign covariance geometry.

\begin{assumption}[SSCM metric stability and separation]
\label{as:theory_A3}
There exist a constant $\eta\in[0,1)$ and events $\mathcal M_n$ with $\Pr(\mathcal M_n)\to1$ such that, on $\mathcal M_n$, every inverse-SSCM metric used along the recovery path satisfies
\begin{equation}
(1-\eta)\mathbf{A}_\star\preceq \mathbf{A}_t\preceq (1+\eta)\mathbf{A}_\star.
\label{eq:theory_metric_stability}
\end{equation}
In addition, for all sufficiently large $n$,
\begin{equation*}
d>4a_n,
\end{equation*}
and
\begin{equation}
d_{\rm sgn}>
\frac{2a_n}{\sqrt{\lambda}}
\left\{1+\sqrt{\frac{1+\eta}{1-\eta}}\right\}.
\label{eq:theory_sep_sscm}
\end{equation}
\end{assumption}

\begin{remark}[Metric form of the SM--SSCM separation]
The spatial-sign covariance enters the theorem through $\mathbf{A}_\star=(\mathbf{\Gamma}_{\rm sgn}+\lambda \mathbf{I}_p)^{-1}$ and the separation $d_{\rm sgn}$. A convenient sufficient condition for \eqref{eq:theory_metric_stability} is the relative SSCM bound
\[
\max_t
\left\|\mathbf{A}_\star^{1/2}(\widehat{\mathbf{\Sigma}}_t-\mathbf{\Sigma}_\star)\mathbf{A}_\star^{1/2}\right\|_{\rm op}
\le \rho_\eta,
\qquad
\rho_\eta:=\frac{\eta}{1+\eta},
\]
where the maximum is over the assignment steps along the recovery path. Indeed,
\[
\left\|\mathbf{A}_\star^{1/2}(\widehat{\mathbf{\Sigma}}_t-\mathbf{\Sigma}_\star)\mathbf{A}_\star^{1/2}\right\|_{\rm op}
\le\rho
\quad\Longrightarrow\quad
(1+\rho)^{-1}\mathbf{A}_\star\preceq \mathbf{A}_t\preceq(1-\rho)^{-1}\mathbf{A}_\star,
\]
and $\rho\le\rho_\eta$ implies \eqref{eq:theory_metric_stability}. Independently of this relative approximation, the spatial-sign normalization gives the deterministic bounds
\[
0\preceq \widehat{\mathbf{S}}_t\preceq \mathbf{I}_p,
\qquad
\lambda \mathbf{I}_p\preceq\widehat{\mathbf{\Sigma}}_t\preceq(1+\lambda)\mathbf{I}_p,
\qquad
(1+\lambda)^{-1}\mathbf{I}_p\preceq \mathbf{A}_t\preceq\lambda^{-1}\mathbf{I}_p.
\]
These deterministic bounds yield a conservative Euclidean sufficient condition, but the theorem itself uses the sharper inverse-SSCM geometry through $d_{\rm sgn}$ and \eqref{eq:theory_metric_stability}.
\end{remark}

\begin{theorem}[Clustering consistency of SM--SSCM]
\label{thm:theory_smsscm}
Suppose Assumptions~\ref{as:theory_A1}--\ref{as:theory_A3} hold. Consider the max--min initialized SM--SSCM iteration whose assignment step uses the regularized inverse-SSCM metric $\mathbf{A}_t=\widehat{\mathbf{\Sigma}}_t^{-1}$, whose center update uses a deterministic selection of spatial medians, and whose tie-breaking rule is deterministic. Then the resulting output partition is clustering consistent in the sense of \eqref{eq:theory_clust_consistency}. Moreover, with probability tending to one, the algorithm terminates in finitely many iterations at a fixed point corresponding to the true partition up to a permutation.
\end{theorem}

Theorem~\ref{thm:theory_smsscm} shows that, under suitable separation in the population spatial-sign covariance geometry and stability of the empirical SSCM metric, the SSCM-based assignment rule combined with spatial-median center updates is asymptotically able to recover the true clustering structure. This provides the theoretical basis for SM--SSCM before sparsification is introduced.

\subsection{Feature-selection and clustering consistency of Sparse--SM}
\label{subsec:theory_sparse}

We next consider the sparse extension based on hard-threshold exclusion. Given current centers $\bm{m}_1,\ldots,\bm{m}_K\in\mathbb{R}^p$, define the coordinate-wise separation score
\begin{equation*}
s_j(\bm{m}_1,\ldots,\bm{m}_K)
=
\sum_{k=1}^K |m_{kj}-\bar m_j|,
\qquad
\bar m_j=\frac{1}{K}\sum_{k=1}^K m_{kj},
\end{equation*}
and the corresponding active set at threshold $\tau$ by
\begin{equation*}
S(\tau):=\{j:s_j(\bm{m}_1,\ldots,\bm{m}_K)\ge \tau\}.
\end{equation*}

At the population level, define
\begin{equation*}
s_j^\star
=
\sum_{k=1}^K |\mu_{kj}-\bar\mu_j|,
\qquad
\bar\mu_j=\frac{1}{K}\sum_{k=1}^K \mu_{kj},
\qquad
S_0:=\{j:s_j^\star>0\}.
\end{equation*}
Thus, $S_0$ is the set of truly cluster-informative coordinates.

\begin{assumption}[Sparsity and signal strength]
\label{as:theory_A4}
Assume $S_0\ne\varnothing$. For every $j\notin S_0$, we have $\mu_{1j}=\cdots=\mu_{Kj}$, and
\[
s_{\min}:=\min_{j\in S_0}s_j^\star>0.
\]
\end{assumption}

Define the restricted separation
\begin{equation*}
d_0:=\min_{r\neq s}\|(\bm{\mu}_r)_{S_0}-(\bm{\mu}_s)_{S_0}\|,
\end{equation*}
where $(\bm{\mu}_r)_{S_0}$ denotes the restriction of $\bm{\mu}_r$ to the active coordinates in $S_0$.

\begin{assumption}[Restricted separation and threshold margin]
\label{as:theory_A5}
For all sufficiently large $n$,
\begin{equation*}
d_0>4a_n,
\end{equation*}
and the deterministic threshold sequence $\tau_n$ satisfies
\begin{equation}
2Ka_n<\tau_n<s_{\min}-2Ka_n.
\label{eq:theory_tau_margin}
\end{equation}
\end{assumption}

In the present theoretical analysis, the threshold sequence $\tau_n$ is deterministic. The data-driven Gap-selected threshold used in the empirical section is not analyzed here.

The first result below shows that hard-threshold exclusion consistently identifies the true informative coordinates once the estimated centers are sufficiently close to the population centers. The proof is based on a uniform perturbation bound for the coordinate-wise separation scores and is deferred to Appendix~\ref{app:proofs}.

\begin{theorem}[Feature-selection consistency of hard-threshold exclusion]
\label{thm:theory_feature}
Suppose Assumptions~\ref{as:theory_A1}--\ref{as:theory_A2} and~\ref{as:theory_A4}--\ref{as:theory_A5} hold. Consider one hard-thresholding step at threshold $\tau_n$ based on centers $\{\bm{m}_k\}_{k=1}^K$ satisfying
\begin{equation}
\max_{1\le k\le K}\|\bm{m}_k-\bm{\mu}_{\sigma(k)}\|\le a_n
\quad
\text{for some permutation }\sigma.
\label{eq:theory_center_close}
\end{equation}
Then the resulting active set equals the population active set $S_0$, that is,
\[
S(\tau_n)=S_0.
\]
Consequently, if the event in \eqref{eq:theory_center_close} holds with probability tending to one, then
\[
\Pr\big(S(\tau_n)=S_0\big)\to 1.
\]
\end{theorem}

We next show that, under the same set of conditions, the sparse hard-threshold procedure preserves both correct variable selection and correct clustering throughout the iterative updates.

\begin{theorem}[Clustering consistency of Sparse--SM]
\label{thm:theory_sparse}
Suppose Assumptions~\ref{as:theory_A1}--\ref{as:theory_A2} and~\ref{as:theory_A4}--\ref{as:theory_A5} hold. Consider the max--min initialized hard-threshold sparse $K$-spatial-median procedure at threshold $\tau_n$, with deterministic spatial-median selection and deterministic tie-breaking. Then, with probability tending to one, the initial hard-thresholding step selects $S_0$, the first assignment on $S_0$ recovers the true partition up to a permutation, and both the active set and the partition remain correct at all subsequent iterations. Moreover, the algorithm terminates in finitely many iterations at a fixed point with active set $S_0$ and the true partition up to a permutation.
\end{theorem}

Theorems~\ref{thm:theory_feature} and~\ref{thm:theory_sparse} show that Sparse--SM achieves both support recovery and clustering recovery under explicit hard-threshold exclusion, provided that the true signal is sufficiently separated from the noise level induced by within-cluster fluctuations.

\subsection{Explicit sufficient conditions}
\label{subsec:theory_explicit}

The preceding assumptions can be written directly in terms of the deviation radius $a_n$. Under the sub-Weibull condition (T$_\alpha$), define
\[
B_{\alpha,n}:=
\sqrt{\lambda_{\max}}\,\kappa\{\log(2n^{1+\varepsilon})\}^{1/\alpha}.
\]
Then $a_n=B_{\alpha,n}$, and Assumptions~\ref{as:theory_A3} and~\ref{as:theory_A5} are implied by
\[
d>4B_{\alpha,n},
\]
\[
d_{\rm sgn}>
\frac{2B_{\alpha,n}}{\sqrt{\lambda}}
\left\{1+\sqrt{\frac{1+\eta}{1-\eta}}\right\},
\]
\[
d_0>4B_{\alpha,n},
\qquad
s_{\min}>4K B_{\alpha,n},
\]
together with the SSCM metric-stability event \eqref{eq:theory_metric_stability} and any threshold sequence $\tau_n$ satisfying \eqref{eq:theory_tau_margin}.

Under the polynomial-tail condition (T$_\nu$), define
\[
B_{\nu,n}:=\sqrt{\lambda_{\max}}(Cn^{1+\varepsilon})^{1/\nu}.
\]
For all sufficiently large $n$, $a_n=B_{\nu,n}$ is in the range where the polynomial tail bound applies, and the corresponding sufficient conditions are
\[
d>4B_{\nu,n},
\]
\[
d_{\rm sgn}>
\frac{2B_{\nu,n}}{\sqrt{\lambda}}
\left\{1+\sqrt{\frac{1+\eta}{1-\eta}}\right\},
\]
\[
d_0>4B_{\nu,n},
\qquad
s_{\min}>4K B_{\nu,n},
\]
again together with \eqref{eq:theory_metric_stability} and any $\tau_n$ satisfying \eqref{eq:theory_tau_margin}.

These conditions separate the three deterministic requirements used in the proof: Euclidean separation for max--min seeding, inverse-SSCM separation for SM--SSCM assignment, and feature-level signal for hard-threshold support recovery. Heavier tails enlarge $a_n$ and therefore require stronger separation and stronger feature-level signal for exact recovery.

\medskip
The proofs of Lemma~\ref{lem:theory_En} and Theorems~\ref{thm:theory_smsscm}--\ref{thm:theory_sparse}, together with the intermediate technical lemmas, are deferred to Appendix~\ref{app:proofs}.

\section{Experiments} \label{sec:experiments}

In this section, we evaluate the proposed clustering method through simulation studies. We first compare performance under light-tailed, heavy-tailed, and contaminated distributions, then examine scalability as the ambient dimension increases, and finally evaluate robustness under row-wise and cell-wise contamination.
Since the proposed clustering procedure is built upon the $K$-spatial-median framework
with an explicit hard-thresholding variable selection mechanism, we compare our method
with several representative clustering algorithms that are closely related in terms of
robustness and sparsity. Specifically, we include the classical $K$-means algorithm,
its robust counterpart $K$-median, as well as their sparse extensions,
namely Sparse $K$-means (SKMeans; \citeauthor{witten2010sparse}, \citeyear{witten2010sparse})
and Sparse $K$-median (SKMedian).


\subsection{Data Generating Process and Baselines}\label{DGP}

In this section, we describe the general setup for generating the simulated datasets used in the simulation studies. Each simulated dataset uses $n_0$ observations per cluster, with $K=3$ clusters, resulting in a total of $N = K\times n_0$ observations. The features are designed to include both informative and uninformative dimensions ($p$) to test the clustering method's ability to separate the groups. Specifically, we consider three types of distributions for the data:

\begin{itemize}
    \item \textbf{Multivariate Gaussian distribution $\mathcal{N}(\bm{\mu}_k,\mathbf{\Sigma})$}: Each cluster is generated from a multivariate Gaussian distribution with mean vectors $\bm{\mu}_k$ and covariance matrix $\mathbf{\Sigma}$.
    \item \textbf{Multivariate $t_3$ distribution}: Each cluster is generated from a multivariate $t$ model with location $\bm{\mu}_k$ to introduce heavy-tailed behavior.
    \item \textbf{Contaminated Gaussian Mixture Model $MN_{0.9}$}:
For each cluster $k = 1,2,3$, observations are generated as
\[
\bm{X}_{ik} = \bm{\mu}_k + B_i \bm{Z}_{ik},
\]
where
\[
\bm{Z}_{ik} \sim \mathcal{N}(\bm{0},\mathbf{\Sigma}),
\]
and the scale variable $B_i$ follows
\[
B_i =
\begin{cases}
1, & \text{with probability } 0.9, \\
3, & \text{with probability } 0.1.
\end{cases}
\]
This corresponds to a two-component scale mixture of Gaussian distributions within each cluster, where 10\% of the observations are inflated to generate outliers.

\end{itemize}

The covariance matrix $\mathbf{\Sigma}$ is of AR(1) type, with entries
\begin{align*}
(\mathbf{\Sigma})_{ij} = 0.9^{|i-j|}.
\end{align*}
The mean vectors \(\bm{\mu}_1\), \(\bm{\mu}_2\), and \(\bm{\mu}_3\) are specified to reflect sparsity: only the first \(s_p=p/20\) coordinates are shifted, whereas the remaining coordinates are set to zero. For example, when \(p=200\), this gives \(s_p=10\). Specifically, we take
\[
\bm{\mu}_1=\bm{0}, \quad
\bm{\mu}_2=(3\cdot \bm{1}_{s_p}, 0, \ldots, 0), \quad
\bm{\mu}_3=(-3\cdot \bm{1}_{s_p}, 0, \ldots, 0),
\]
where \(\bm{1}_{s_p}\) denotes a vector of ones of length \(s_p\).

We use this data generation process to evaluate the ability of the proposed Sparse $K$-spatial-median method to correctly recover the underlying cluster structure under both Gaussian and heavy-tailed distributions with sparse mean differences. We compare the proposed methods with several widely used clustering approaches:

\begin{itemize}

\item \textbf{SM--SSCM:} the proposed robust clustering algorithm based on spatial medians and the spatial sign covariance matrix (Algorithm~\ref{alg:robust}).

\item \textbf{Sparse--SM:} the proposed sparse clustering algorithm based on spatial medians with dimension exclusion and Gap-based tuning (Algorithm~\ref{alg:sparse}).

\item \textbf{$K$-means:} the classical Lloyd's algorithm based on Euclidean distance.

\item \textbf{$K$-medians:} the classical clustering algorithm based on componentwise medians and $\ell_1$ distance.

\item \textbf{Sparse $K$-means:} the sparse clustering method of \cite{witten2010sparse}, implemented in the \texttt{sparcl} package.

\item \textbf{Sparse $K$-median:} the sparse clustering method based on componentwise medians with dimension exclusion.

\end{itemize}

Clustering performance is evaluated using the Adjusted Rand Index (ARI) \citep{hubert1985ari}. For each method, results are averaged over repeated replications.
Specifically, we conduct two sets of experiments. First, we evaluate clustering performance under Gaussian and heavy-tailed distributions. Within this experiment, we consider both fixed-dimensional settings and increasing-dimensional settings to assess scalability in high-dimensional regimes. Second, we examine the robustness of the clustering methods in the presence of outliers.

All methods are implemented in \texttt{R}. For algorithms requiring initialization, we use 20 random restarts and report the solution with the best objective value. For sparse methods, the thresholding parameter is selected by the Gap criterion described in Section~\ref{subsec:tau_selection}.
For SM--SSCM, the common assignment metric is computed from a regularized SSCM estimate. In implementation, a small ridge term is added to ensure invertibility, and optional banding may be applied in high dimensions. When an empty cluster occurs, a random reassignment step is used as a practical safeguard to continue the iteration.






\subsection{Performance under Different Distributions}
In this experiment, we fix $p = 200$, $n_0 = 100$ per cluster, and $K = 3$. Data are generated as described in Section~\ref{DGP} under multivariate Gaussian and $t_3$ distributions to assess performance under light- and heavy-tailed settings. Clustering accuracy is measured by the average ARI over 100 replications.

\begin{table}[ht]
\caption{Average ARI over 100 replications under different distributions.
Values are reported as mean (standard deviation).}
\label{tab:ari}
\centering
\begin{tabular}{lccc}
\toprule
Method & Gaussian & $t_{3}$& $MN_{0.9}$ \\
\midrule
SM--SSCM &0.730(0.144) & 0.401(0.062)& 0.473(0.119)\\
Sparse--SM &0.843(0.035) &0.606(0.152)&0.756(0.049) \\
$K$-means  &0.784(0.091) & 0.393(0.078)&0.475(0.134) \\
$K$-medians& 0.683(0.181) &0.383(0.191)&0.595(0.161)  \\
Sparse $K$-means  &0.843(0.037) &0.466(0.183)&0.755(0.049) \\
Sparse $K$-median & 0.844(0.034) &0.678(0.079)&0.744(0.101) \\
\bottomrule
\end{tabular}
\end{table}

As shown in Table~\ref{tab:ari},
Sparse--SM achieves performance comparable to the best sparse competitors under the Gaussian model, remains clearly superior to the non-sparse methods under the heavy-tailed \(t_3\) model, and attains the highest average ARI under the contaminated Gaussian model \(MN_{0.9}\). These results suggest that the proposed sparse spatial-median framework delivers strong and stable performance across light-tailed, heavy-tailed, and contaminated settings.

We next examine the scalability of the competing methods under the Gaussian model. In this experiment, we fix the number of observations per cluster at \(n_0=100\) and let the ambient dimension vary over
\[
p \in \{50,100,200,400\}.
\]
The number of informative features is set to \(s_p=p/20\), so that the sparsity ratio remains fixed as \(p\) increases. Data are generated from the multivariate Gaussian distribution with the same mean and covariance structure described in Section~\ref{DGP}. For each setting, results are averaged over 100 replications.

\begin{figure}[htbp]
    \centering
    \includegraphics[width=0.7\textwidth]{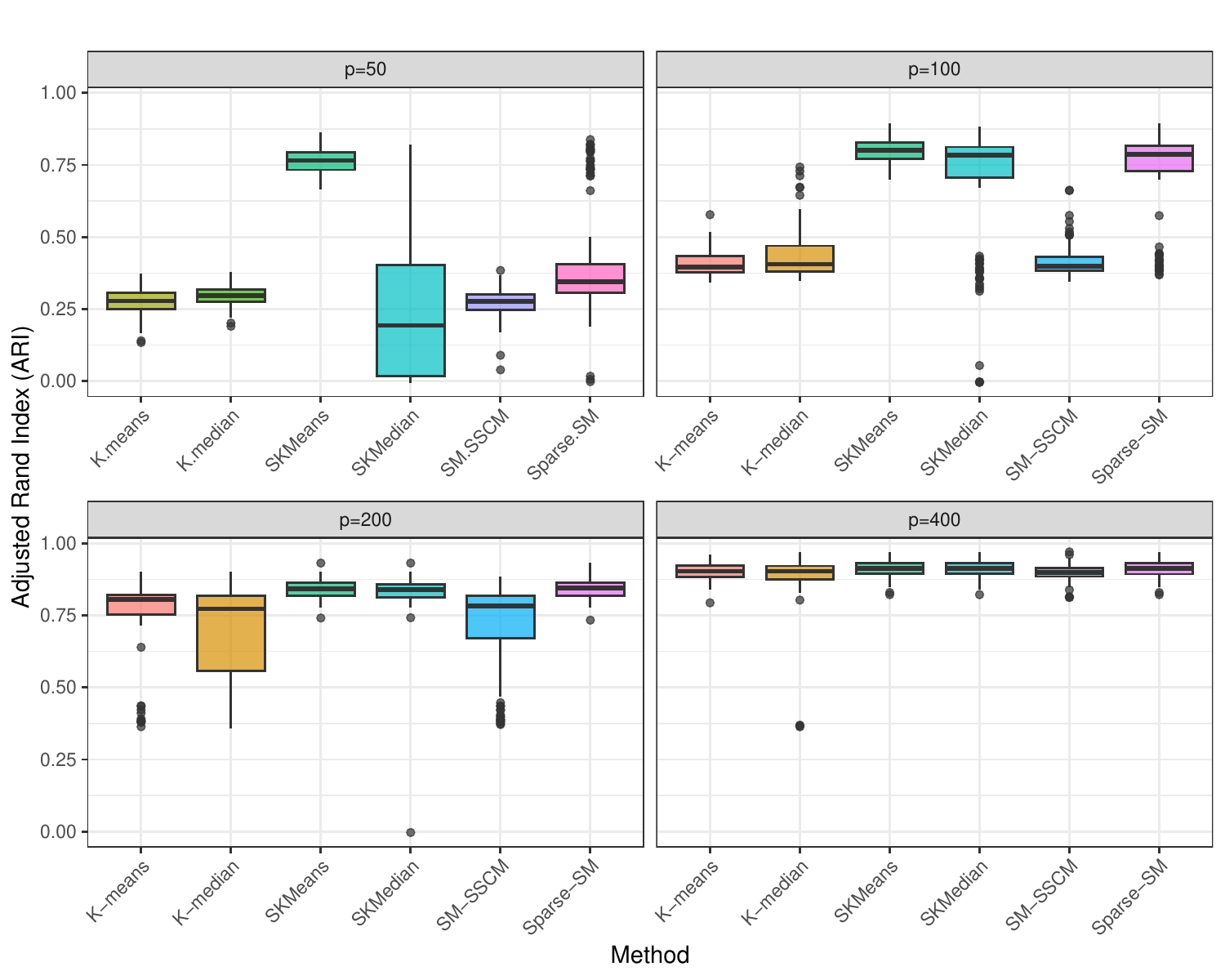}
    \caption{ARI values of different methods across varying dimensions \(p\in\{50,100,200,400\}\), with \(n_0=100\) observations per cluster.}
    \label{fig:ari_distribution}
\end{figure}

Figure~\ref{fig:ari_distribution} shows that the sparse methods remain substantially more stable than the non-sparse alternatives as the dimension increases, confirming the importance of excluding irrelevant features under the present sparse signal design. At \(p=50\), Sparse--SM already improves markedly upon SM--SSCM, \(K\)-means, and \(K\)-medians, although its performance is still below that of sparse \(K\)-means. As \(p\) increases to 100 and 200, Sparse--SM becomes one of the strongest performers and remains competitive with the best sparse baselines. When \(p=400\), all methods improve because the number of informative coordinates also increases with \(p\), but Sparse--SM continues to exhibit strong and stable performance. 


\subsection{Robustness against Outliers}

To evaluate robustness under different contamination mechanisms, we consider both row-wise and cell-wise outlier models. These two settings represent distinct forms of data corruption: row-wise contamination replaces entire observations by gross outliers, whereas cell-wise contamination corrupts only a subset of individual entries and is therefore more localized but potentially more challenging in high dimensions.

\textit{Clean data} are generated from the Gaussian mixture model described in Section~\ref{DGP}. Specifically, for each cluster \(k=1,2,3\),
\[
\bm{X}_i \sim \mathcal{N}(\bm{\mu}_k,\mathbf{\Sigma}),
\]
where \(\mathbf{\Sigma}\) has autoregressive entries \((\mathbf{\Sigma})_{ij}=0.9^{|i-j|}\). The dimension is fixed at \(p=200\), with \(p_{\mathrm{inf}}=p/20\) informative variables and \(p_{\mathrm{noise}}=p-p_{\mathrm{inf}}\) noise variables. The cluster means differ only on the informative coordinates.

For \textit{row-wise contamination}, each observation is independently replaced by a noise vector with probability \(\varepsilon\):
\[
\bm{X}_i=
\begin{cases}
\bm{X}_i^{\mathrm{clean}}, & \text{with probability } 1-\varepsilon,\\
\bm{W}_i, & \text{with probability } \varepsilon,
\end{cases}
\]
where \(\bm{W}_i\sim \mathcal{N}(\bm{0},\sigma^2 \mathbf{I}_p)\) with \(\sigma=5\). We consider \(\varepsilon\in\{0,0.05,0.10,0.15,0.20\}\).

For \textit{cell-wise contamination}, each entry is independently replaced by noise according to
\[
X_{ij}=(1-\Delta_{ij})X_{ij}^{\mathrm{clean}}+\Delta_{ij}W_{ij},
\]
where \(\Delta_{ij}\sim\mathrm{Bernoulli}(\varepsilon_j)\). For informative variables, \(\varepsilon_j=\varepsilon\), whereas for noise variables we fix \(\varepsilon_j=0.10\). The contamination noise satisfies \(W_{ij}\sim \mathcal{N}(0,3^2)\). Since cell-wise corruption accumulates rapidly in high dimensions, we consider \(\varepsilon\in\{0,0.02,0.03,0.05\}\). For each contamination level, results are averaged over 100 replications.

\begin{figure}[htbp]
    \centering
    \begin{subfigure}{0.48\textwidth}
        \centering
        \includegraphics[width=\textwidth]{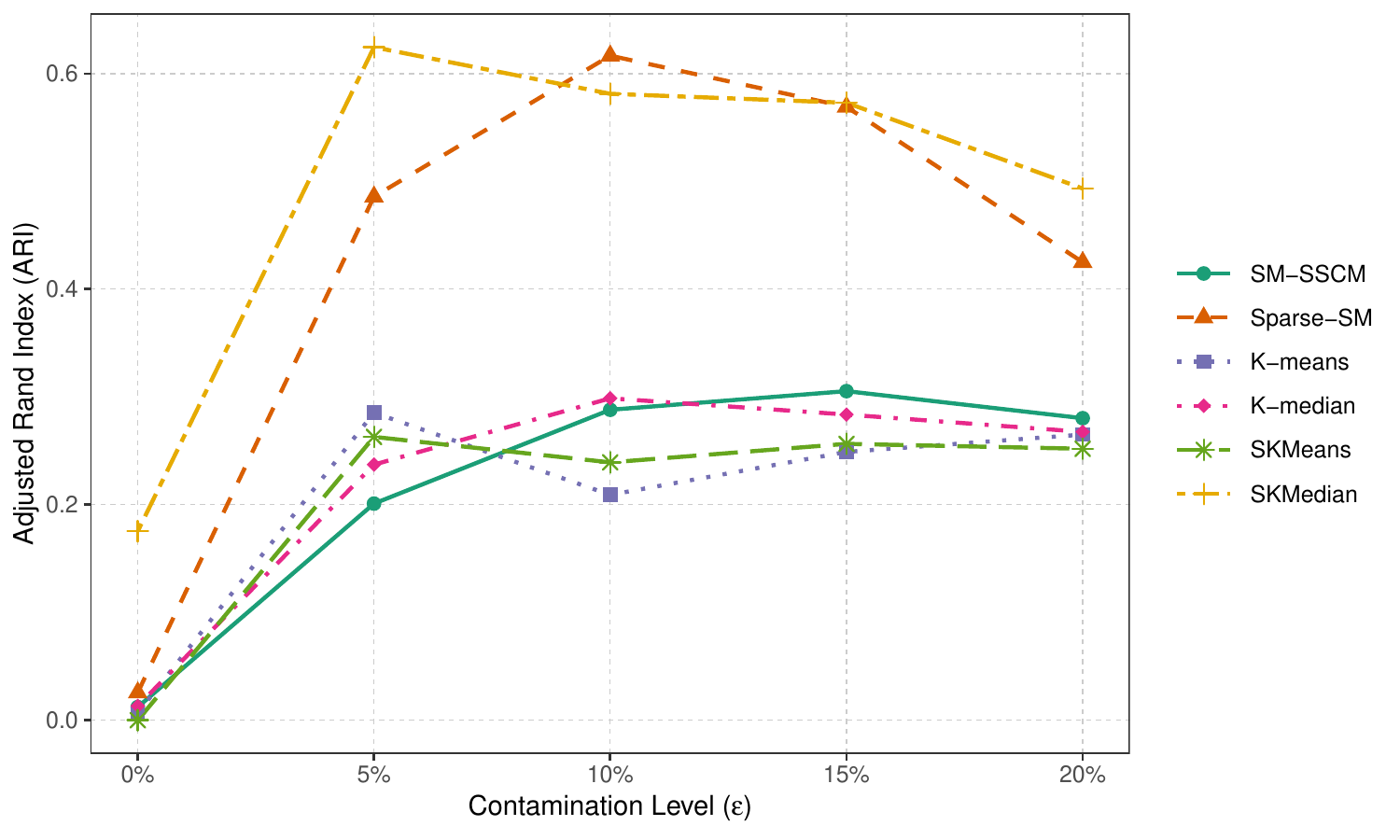}
        \caption{Row-wise contamination}
    \end{subfigure}
    \hfill
    \begin{subfigure}{0.48\textwidth}
        \centering
        \includegraphics[width=\textwidth]{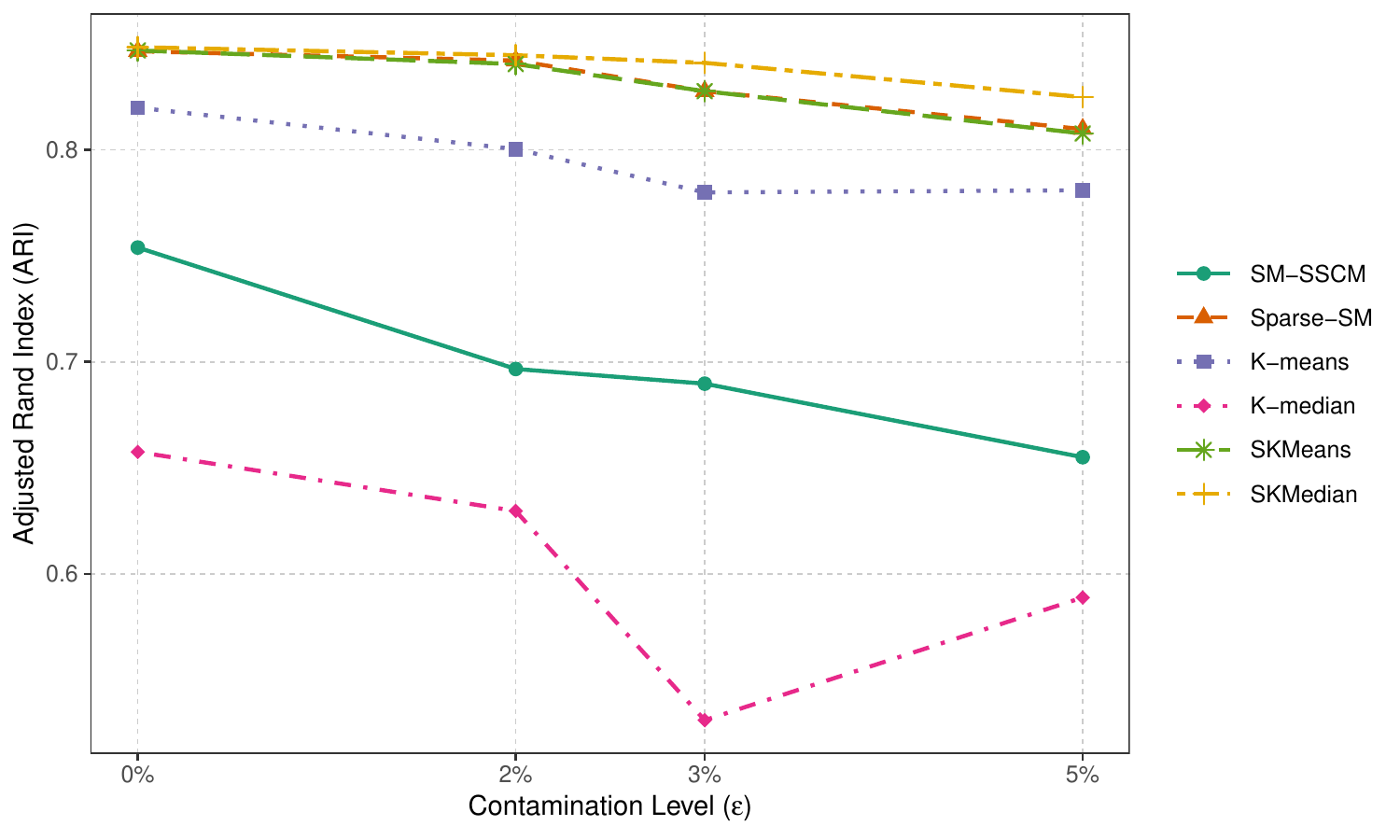}
        \caption{Variable-wise (cell-wise) contamination}
    \end{subfigure}
    \caption{ARI values of competing methods under different contamination mechanisms.}
    \label{fig:contamination_comparison}
\end{figure}

Figure~\ref{fig:contamination_comparison} shows that the proposed Sparse--SM method remains competitive under both contamination schemes, with a particularly clear advantage under cell-wise corruption.
The row-wise contamination results in panel (a) are more variable across methods and contamination levels. Nevertheless, Sparse--SM performs substantially better than the non-sparse procedures once contamination is introduced and remains competitive with the strongest sparse baselines over a broad range of contamination levels.
In panel (b), Sparse--SM consistently belongs to the top-performing group across all contamination levels and exhibits only mild degradation as the contamination intensity increases. This suggests that combining spatial-median center updates with hard feature exclusion is effective in the presence of localized corruption, where robustness and variable selection are both essential.


\section{Real-Data Applications}\label{sec:realdata}

\subsection{Mice protein data}

We illustrate the proposed methods using the mice protein expression dataset, which contains measurements on 77 proteins from the nuclear fraction of the cortex for control and Down syndrome (Ts65Dn) mice. Each mouse was measured 15 times. In addition to mouse ID, the dataset records genotype (Control or Trisomy), treatment (Memantine or Saline), behavior (Context-Shock or Shock-Context), and the resulting class labels formed by these factors. The data analyzed here are a subset of those considered in \cite{NEURIPS2020_735ddec1}.

The raw data were first preprocessed by removing entries with missing genotype information and excluding mice with more than 50\% missing values across the protein measurements.
Missing protein values are imputed by the mean within the corresponding mouse class. For each mouse, the 15 repeated measurements are then aggregated by taking the protein-wise average, producing one observation per mouse. Finally, the resulting feature matrix is normalized to make protein expression levels comparable across variables. The subsequent analysis is carried out separately for the Control and Trisomic groups.

\begin{figure}[htbp]
    \centering
    \begin{minipage}{0.48\textwidth}
        \centering
        \includegraphics[width=\linewidth]{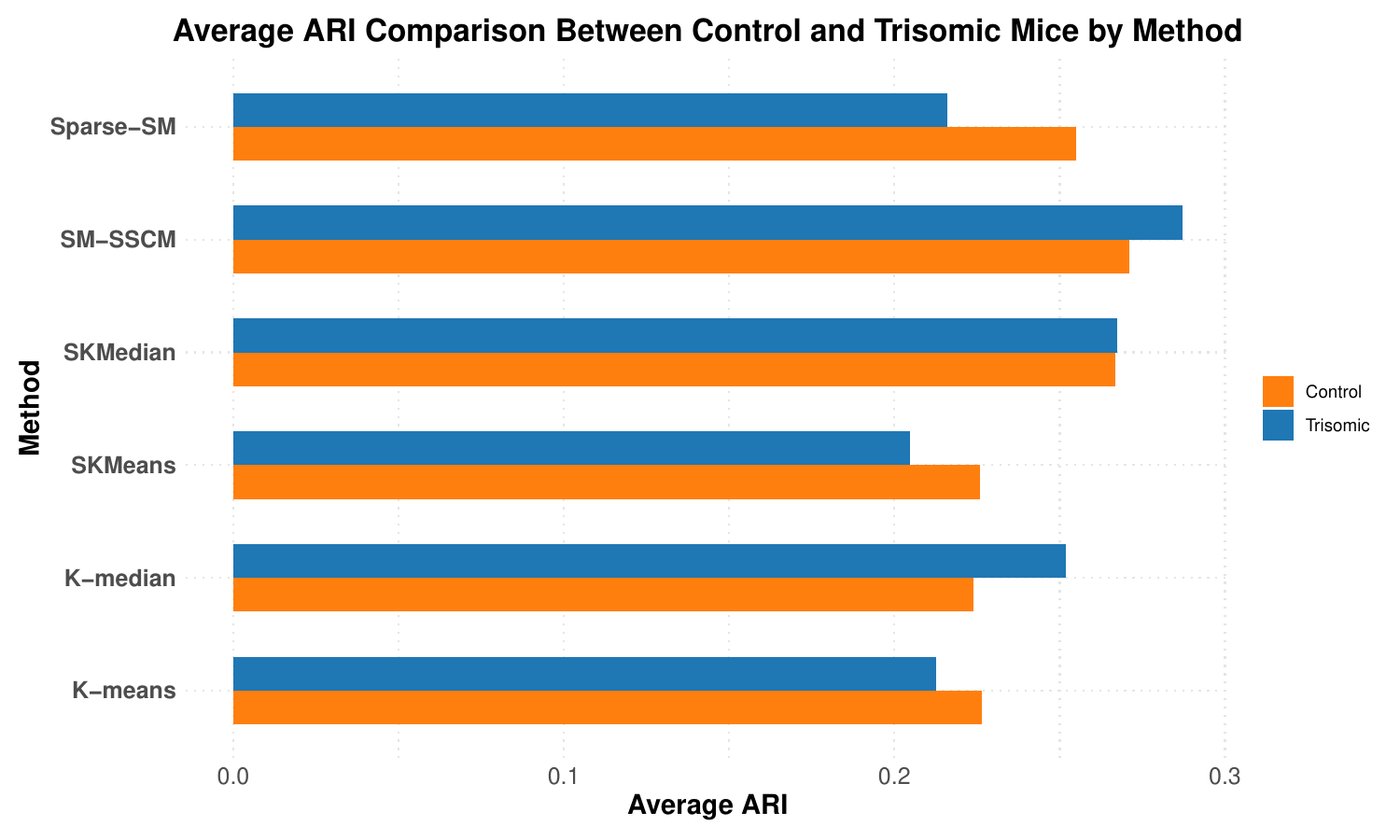}
    \end{minipage}\hfill
    \begin{minipage}{0.48\textwidth}
        \centering
        \includegraphics[width=\linewidth]{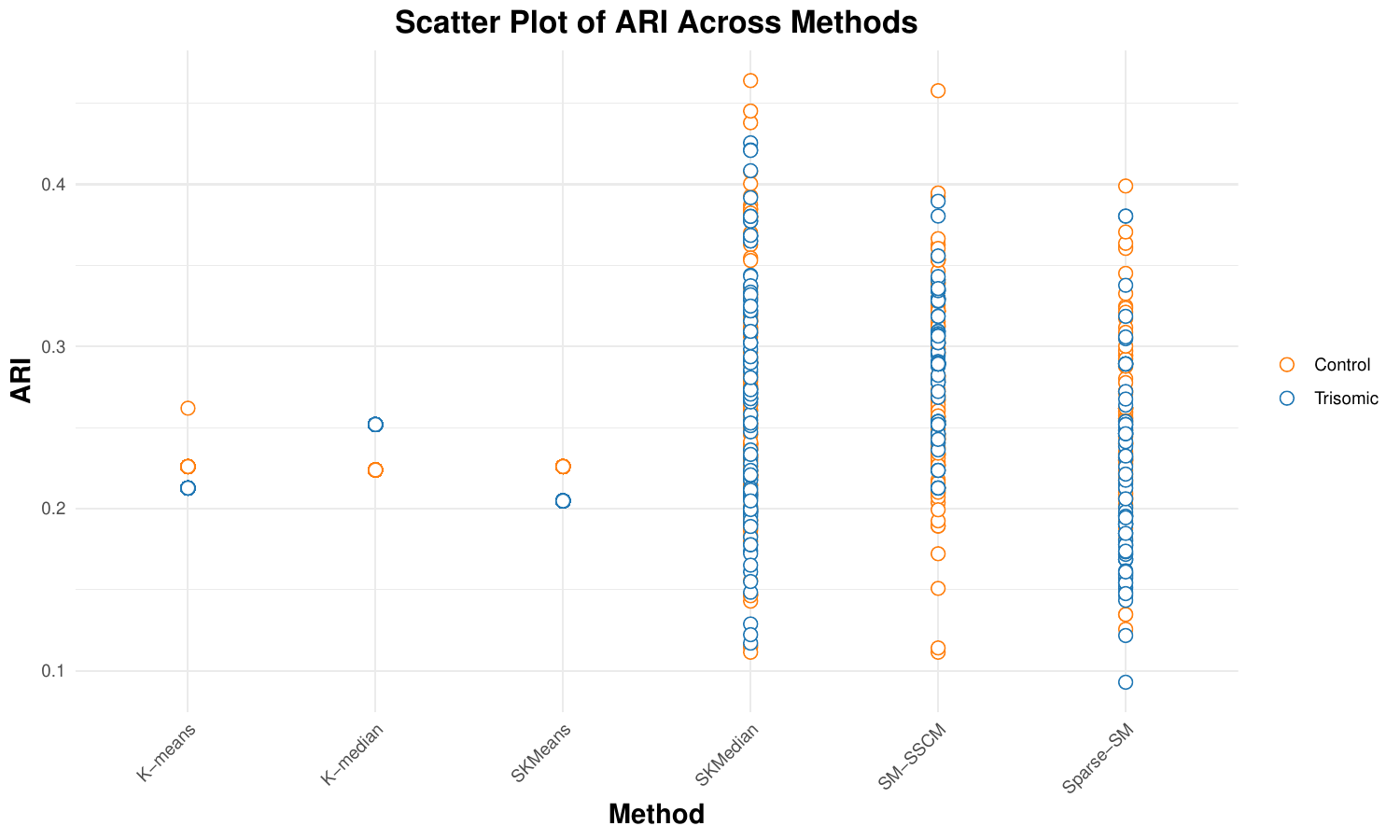}
    \end{minipage}
    \caption{Comparison of clustering performance on the mice protein data for the Control and Trisomic groups. The left panel reports the average ARI for each method, and the right panel displays the corresponding ARI values across repeated runs.}
    \label{fig:mice_protein_results}
\end{figure}

Figure~\ref{fig:mice_protein_results} shows that the proposed methods are competitive on both genotype groups. In terms of average ARI, SM--SSCM performs best on the Trisomic group and remains among the strongest methods on the Control group, while Sparse--SM achieves the highest average ARI on the Control group and also performs strongly on the Trisomic group. The scatter plot further indicates that these methods are not only competitive in average performance but also relatively stable across repeated runs. By contrast, the classical methods, especially \(K\)-means and \(K\)-medians, yield clearly lower ARI values.
These findings suggest that incorporating either SSCM-based metric adaptation or sparse dimension exclusion can substantially improve clustering performance in practice.

\subsection{Benchmark datasets}

We further evaluate the proposed methods on 8 benchmark datasets from the KEEL and UCI repositories. These datasets cover a range of sample sizes, dimensions, and numbers of clusters, and therefore provide a broader assessment of practical performance beyond the synthetic settings. Their basic information and sources are summarized in Table~\ref{tab:benchmark_des}. For each dataset, we run 20 independent trials for every method.

\begin{table}[htbp]
\centering
\caption{Average ARI values of the proposed and competing methods on the benchmark datasets. Standard deviations are reported in parentheses below the corresponding means. The best result in each column is shown in bold. ``MoLi'' denotes the Movement Libras dataset.}
\label{tab:ari_bench}
\setlength{\tabcolsep}{4.5pt}
\renewcommand{\arraystretch}{1.1}
\footnotesize
\begin{tabular}{l
                S[table-format=1.3]
                S[table-format=1.3]
                S[table-format=1.3]
                S[table-format=1.3]
                S[table-format=1.3]
                S[table-format=1.3]
                S[table-format=1.3]
                S[table-format=1.3]}
\toprule
\textbf{Method} &
\textbf{Newthyroid} &
\textbf{Iris} &
\textbf{WDBC} &
\textbf{Zoo} &
\textbf{MoLi} &
\textbf{Ecoli} &
\textbf{Wine} &
\textbf{Glass} \\
\midrule
SM--SSCM   & {0.524} & {\textbf{0.677}} & {0.712} & {\textbf{0.842}} & {0.317} & {0.361} & {\textbf{0.914}} & {0.244}  \\
           & {(0.181)}& {(0.006)}& {(0.009)}& {(0.062)}& {(0.012)}& {(0.076)}& {(0.054)}& {(0.008)}\\
Sparse--SM & {0.598} & {0.610} & {\textbf{0.713}} & {0.162} & {\textbf{0.323}} & {0.386} & {0.757} & {0.224} \\
           & {(0.009)}& {(0.013)}& {(0.009)}& {(0.186)}& {(0.015)}& {(0.057)}& {(0.063)}& {(0.029)}\\
$K$-means  & {0.358} & {0.580} & {0.523} & {0.815} & {0.312} & {\textbf{0.427}} & {0.837} & {\textbf{0.274}} \\
           & {(0.029)}& {(0.096)}& {(0.008)}& {(0.091)}& {(0.010)}& {(0.012)}& {(0.011)}& {(0.003)}\\
$K$-median & {0.624} & {0.609} & {0.669} & {0.521} & {0.319} & {0.360} & {0.804} & {0.197} \\
           & {(0.000)}& {(0.000)}& {(0.000)}& {(0.064)}& {(0.000)}& {(0.000)}& {(0.000)}& {(0.000)}\\
SKMeans    & {0.334} & {0.526} & {0.327} & {0.704} & {0.309} & {0.326} & {0.607} & {0.233} \\
           & {(0.000)}& {(0.000)}& {(0.000)}& {(0.010)}& {(0.007)}& {(0.046)}& {(0.036)}& {(0.006)}\\
SKMedian   & {\textbf{0.701}} & {0.658} & {0.705} & {0.316} & {0.260} & {0.301} & {0.896} & {0.226}\\
           & {(0.000)}& {(0.014)}& {(0.007)}& {(0.076)}& {(0.018)}& {(0.042)}& {(0.075)}& {(0.028)}\\
\bottomrule
\end{tabular}
\end{table}

\begin{table}[htbp]
\centering
\caption{Source and basic description of the benchmark datasets.}
\label{tab:benchmark_des}
\begin{tabular}{l l c c c}
\toprule
\textbf{Dataset} & \textbf{Source} & \textbf{$K$} & \textbf{$n$} & \textbf{$p$} \\
\midrule
Newthyroid      & \href{https://sci2s.ugr.es/keel/category.php?cat=clas}{Keel Repository} & 3  & 215 & 5  \\
Iris            & \href{https://sci2s.ugr.es/keel/category.php?cat=clas}{Keel Repository} & 3  & 150 & 4  \\
WDBC            & \href{https://archive.ics.uci.edu/dataset/17/breast+cancer+wisconsin+diagnostic}{UCI Repository} & 2  & 569 & 30 \\
Zoo             & \href{https://sci2s.ugr.es/keel/category.php?cat=clas}{Keel Repository} & 7  & 101 & 16 \\
Movement Libras & \href{https://sci2s.ugr.es/keel/category.php?cat=clas}{Keel Repository} & 15 & 360 & 90 \\
Ecoli           & \href{https://sci2s.ugr.es/keel/category.php?cat=clas}{Keel Repository} & 8  & 336 & 7  \\
Wine            & \href{https://sci2s.ugr.es/keel/category.php?cat=clas}{Keel Repository}   & 3  & 178 & 13 \\
Glass           & \href{https://archive.ics.uci.edu/dataset/42/glass+identification}{UCI Repository}  & 6  & 214 & 9  \\
\bottomrule
\end{tabular}
\end{table}

Table~\ref{tab:ari_bench} shows that the proposed methods remain competitive across the benchmark collection. SM--SSCM achieves the best performance on Iris, Zoo, and Wine, while Sparse--SM performs best on WDBC and Movement Libras. On several other datasets, including Newthyroid and Ecoli, the proposed methods remain close to the strongest baselines even when they do not attain the highest average ARI\@. The spatial-median-based framework transfers well to heterogeneous benchmark data, with SM--SSCM being particularly effective when metric adaptation is beneficial and Sparse--SM being advantageous when explicit feature exclusion enhances cluster separation.

\section{Conclusion}\label{sec:conclusion}

In this paper, we proposed a robust clustering framework for high-dimensional data based on spatial medians. The proposed methodology combines spatial-median center updates with a common spatial-sign covariance metric for cluster assignment, and further extends this framework through explicit dimension exclusion to handle sparse high-dimensional signals. A permutation-based Gap criterion was also developed for automatic and data-driven tuning of the sparsity threshold.

The proposed methods were designed to address two central difficulties in modern clustering problems, namely robustness under heavy-tailed or contaminated data and performance degradation caused by a large number of irrelevant variables. In this framework, SM--SSCM improves the assignment step through robust metric adaptation, whereas Sparse--SM further enhances clustering performance in sparse high-dimensional settings by excluding weakly separating coordinates. Theoretical analysis established consistency properties for the proposed approach, while simulation studies and real-data applications demonstrated its competitive performance across a range of settings, especially in the presence of contamination and irrelevant variables.

\section*{Acknowledgments}
Omitted for anonymous submission.

\appendix
\section{Proofs of the Main Consistency Results}
\label{app:proofs}

This appendix proves Lemma~\ref{lem:theory_En} and Theorems~\ref{thm:theory_smsscm}--\ref{thm:theory_sparse}. We use the same notation as in Section~\ref{sec:theory}. The random-restart implementation and the data-driven Gap selector are not analyzed. All nearest-center ties and nonunique spatial medians are resolved by the deterministic rule specified in the theoretical setup.

\subsection{Proof of Lemma~\ref{lem:theory_En}}
\label{app:aux_concentration}

\begin{proof}[Proof of Lemma~\ref{lem:theory_En}]
For every observation, Assumption~\ref{as:theory_A1} gives
\begin{equation}
\|\bm{X}_i-\bm{\mu}_{z_i}\|
=\|\mathbf{\Sigma}_{z_i}^{1/2}(R_i\bm{U}_i)\|
\le \lambda_{\max}^{1/2}R_i\|\bm{U}_i\|
=\lambda_{\max}^{1/2}R_i.
\label{eq:norm_radial_bound_app}
\end{equation}
If $\lambda_{\max}=0$, then $a_n=0$ and \eqref{eq:norm_radial_bound_app} gives $\mathcal E_n$ surely. Hence assume $\lambda_{\max}>0$. Then
\begin{equation*}
\mathcal E_n^c
\subseteq
\left\{\max_{1\le i\le n}R_i>a_n/\sqrt{\lambda_{\max}}\right\},
\qquad
\Pr(\mathcal E_n^c)
\le n\Pr\left(R_1>a_n/\sqrt{\lambda_{\max}}\right).
\end{equation*}
If (T$_\alpha$) holds, then
\begin{align*}
\Pr(\mathcal E_n^c)
&\le 2n\exp\left[-\left\{\frac{\kappa(\log(2n/\delta_n))^{1/\alpha}}{\kappa}\right\}^{\alpha}\right]  \notag\\
&=2n\exp\{-\log(2n/\delta_n)\}
=\delta_n.
\end{align*}
If (T$_\nu$) holds, then for all sufficiently large $n$ such that
\((Cn/\delta_n)^{1/\nu}\ge t_0\),
\begin{align*}
\Pr(\mathcal E_n^c)
&\le nC\left\{(Cn/\delta_n)^{1/\nu}\right\}^{-\nu}
=\delta_n.
\end{align*}
This proves the stated probability bound. \qedhere
\end{proof}

\subsection{Auxiliary deterministic comparisons}
\label{app:deterministic_comparisons}

Let
\[
C_{r,T}:=\{i:z_i=r\},
\qquad
\mathcal F_n:=\left\{\min_{1\le r\le K}|C_{r,T}|\ge1\right\}.
\]
Since \(\Pr(C_{r,T}=\varnothing)=(1-\pi_r)^n\),
\begin{equation}
\Pr(\mathcal F_n^c)
\le \sum_{r=1}^K(1-\pi_r)^n
\le K(1-\pi_{\min})^n
\to0.
\label{eq:Fn_prob_app}
\end{equation}

\begin{lemma}[Max--min seeding hits all true clusters]
\label{lem:seed_hits_all}
On \(\mathcal E_n\cap\mathcal F_n\), if \(d>4a_n\), then Euclidean max--min seeding selects exactly one observation from each true cluster, up to a permutation of labels.
\end{lemma}

\begin{proof}
On \(\mathcal E_n\), for \(i,j\in C_{r,T}\),
\begin{equation}
\|\bm{X}_i-\bm{X}_j\|
\le \|\bm{X}_i-\bm{\mu}_r\|+\|\bm{X}_j-\bm{\mu}_r\|
\le 2a_n,
\label{eq:within_diam_app}
\end{equation}
and for \(i\in C_{r,T}\), \(j\in C_{s,T}\), \(r\ne s\),
\begin{equation}
\|\bm{X}_i-\bm{X}_j\|
\ge \|\bm{\mu}_r-\bm{\mu}_s\|-\|\bm{X}_i-\bm{\mu}_r\|-\|\bm{X}_j-\bm{\mu}_s\|
\ge d-2a_n
>2a_n.
\label{eq:between_gap_app}
\end{equation}
After \(t<K\) seeds have been selected, let \(\mathcal S_t\) be the seed set, let \(\mathcal R_t\) be the represented true labels, and define
\[
D_t(\bm{x}):=\min_{\bm{y}\in\mathcal S_t}\|\bm{x}-\bm{y}\|.
\]
If \(z(\bm{x})\in\mathcal R_t\), then \eqref{eq:within_diam_app} implies \(D_t(\bm{x})\le2a_n\). If \(z(\bm{x})\notin\mathcal R_t\), then \eqref{eq:between_gap_app} implies
\[
D_t(\bm{x})=\min_{\bm{y}\in\mathcal S_t}\|\bm{x}-\bm{y}\|
\ge d-2a_n>2a_n.
\]
Thus
\[
\max_{z(\bm{x})\notin\mathcal R_t}D_t(\bm{x})
>
\max_{z(\bm{x})\in\mathcal R_t}D_t(\bm{x}),
\]
so the next max--min seed is from an unrepresented true cluster. Since \(\mathcal F_n\) guarantees that every true cluster is nonempty, induction over \(t=1,\ldots,K\) proves the claim. \qedhere
\end{proof}

\begin{lemma}[Euclidean exact assignment from close centers]
\label{lem:euclidean_close_exact}
Assume \(\mathcal E_n\) holds. Suppose centers \(\bm{m}_1,\ldots,\bm{m}_K\) satisfy
\begin{equation*}
\max_{1\le k\le K}\|\bm{m}_k-\bm{\mu}_{\sigma(k)}\|\le a_n
\end{equation*}
for some permutation \(\sigma\). If \(d>4a_n\), then the Euclidean nearest-center assignment is exact up to the same permutation.
\end{lemma}

\begin{proof}
Fix \(\bm{X}_i\in C_{r,T}\) and write \(k_r=\sigma^{-1}(r)\). For the correct center,
\begin{equation*}
\|\bm{X}_i-\bm{m}_{k_r}\|
\le \|\bm{X}_i-\bm{\mu}_r\|+\|\bm{m}_{k_r}-\bm{\mu}_r\|
\le 2a_n.
\end{equation*}
For \(s\ne r\), with \(k_s=\sigma^{-1}(s)\),
\begin{equation*}
\|\bm{X}_i-\bm{m}_{k_s}\|
\ge \|\bm{\mu}_r-\bm{\mu}_s\|-\|\bm{X}_i-\bm{\mu}_r\|-\|\bm{m}_{k_s}-\bm{\mu}_s\|
\ge d-2a_n.
\end{equation*}
Since \(d>4a_n\),
\[
\|\bm{X}_i-\bm{m}_{k_r}\|
\le2a_n
<d-2a_n
\le \|\bm{X}_i-\bm{m}_{k_s}\|,
\qquad s\ne r.
\]
The assignment is therefore exact. \qedhere
\end{proof}

\begin{lemma}[Spatial-median update stays inside the true ball]
\label{lem:median_stability}
On \(\mathcal E_n\), if cluster \(k\) is exactly \(C_{\sigma(k),T}\), then every spatial-median minimizer based on that cluster satisfies
\begin{equation}
\|\bm{m}_k-\bm{\mu}_{\sigma(k)}\|\le a_n.
\label{eq:median_close_app}
\end{equation}
\end{lemma}

\begin{proof}
Fix \(r\) and set
\[
\mathcal X_r:=\{\bm{X}_i:i\in C_{r,T}\},
\qquad
B_r:=\{\bm{x}:\|\bm{x}-\bm{\mu}_r\|\le a_n\}.
\]
On \(\mathcal E_n\), \(\mathcal X_r\subset B_r\), and \(B_r\) is convex. Let \(C:=\operatorname{conv}(\mathcal X_r)\). If \(\bm{m}\notin C\) and \(\bm{m}_0=P_C(\bm{m})\), then for all \(\bm{x}\in C\),
\[
(\bm{m}-\bm{m}_0)^\top(\bm{x}-\bm{m}_0)\le0.
\]
Consequently,
\begin{align*}
\|\bm{m}-\bm{x}\|^2
&=\|(\bm{m}-\bm{m}_0)+(\bm{m}_0-\bm{x})\|^2  \notag\\
&=\|\bm{m}-\bm{m}_0\|^2+\|\bm{m}_0-\bm{x}\|^2-2(\bm{m}-\bm{m}_0)^\top(\bm{x}-\bm{m}_0) \notag\\
&\ge \|\bm{m}-\bm{m}_0\|^2+\|\bm{m}_0-\bm{x}\|^2
>\|\bm{m}_0-\bm{x}\|^2.
\end{align*}
Thus
\[
\sum_{\bm{x}\in\mathcal X_r}\|\bm{m}-\bm{x}\|>
\sum_{\bm{x}\in\mathcal X_r}\|\bm{m}_0-\bm{x}\|,
\]
so no spatial-median minimizer lies outside \(C\). Hence
\[
\bm{m}_k\in\operatorname{conv}(\mathcal X_{\sigma(k)})
\subset B_{\sigma(k)},
\]
which gives \eqref{eq:median_close_app}. \qedhere
\end{proof}

\begin{lemma}[SSCM spectral bounds and relative Loewner comparison]
\label{lem:sscm_spectral_app}
For
\[
\widehat{\mathbf{S}}_t:=\frac1n\sum_{i=1}^n \bm{u}_i^{(t)}\bm{u}_i^{(t)\top},
\qquad
\widehat{\mathbf{\Sigma}}_t:=\widehat{\mathbf{S}}_t+\lambda \mathbf{I}_p,
\qquad
\mathbf{A}_t:=\widehat{\mathbf{\Sigma}}_t^{-1},
\]
one has
\begin{equation}
0\preceq \widehat{\mathbf{S}}_t\preceq \mathbf{I}_p,
\qquad
\lambda \mathbf{I}_p\preceq \widehat{\mathbf{\Sigma}}_t\preceq (1+\lambda)\mathbf{I}_p,
\qquad
(1+\lambda)^{-1}\mathbf{I}_p\preceq \mathbf{A}_t\preceq \lambda^{-1}\mathbf{I}_p.
\label{eq:sscm_spectral_app}
\end{equation}
Moreover, if
\begin{equation}
\left\|\mathbf{A}_\star^{1/2}(\widehat{\mathbf{\Sigma}}_t-\mathbf{\Sigma}_\star)\mathbf{A}_\star^{1/2}\right\|_{\rm op}
\le \rho<1,
\label{eq:relative_sscm_event_app}
\end{equation}
then
\begin{equation}
(1+\rho)^{-1}\mathbf{A}_\star
\preceq \mathbf{A}_t
\preceq (1-\rho)^{-1}\mathbf{A}_\star.
\label{eq:relative_to_loewner_app}
\end{equation}
In particular, \eqref{eq:theory_metric_stability} follows from \eqref{eq:relative_sscm_event_app} whenever \(\rho\le \eta/(1+\eta)\).
\end{lemma}

\begin{proof}
For any \(\bm{v}\in\mathbb R^p\), \(\|\bm{u}_i^{(t)}\|\le1\) gives
\[
0\le \bm{v}^\top \bm{u}_i^{(t)}\bm{u}_i^{(t)\top}\bm{v}
=\{\bm{u}_i^{(t)\top}\bm{v}\}^2
\le \|\bm{v}\|^2.
\]
Averaging gives \(0\preceq\widehat{\mathbf{S}}_t\preceq \mathbf{I}_p\). Adding \(\lambda \mathbf{I}_p\) and inverting positive definite matrices gives \eqref{eq:sscm_spectral_app}.

Let
\[
\mathbf{H}_t:=\mathbf{A}_\star^{1/2}\widehat{\mathbf{\Sigma}}_t\mathbf{A}_\star^{1/2}
=\mathbf{I}_p+\mathbf{A}_\star^{1/2}(\widehat{\mathbf{\Sigma}}_t-\mathbf{\Sigma}_\star)\mathbf{A}_\star^{1/2}.
\]
Under \eqref{eq:relative_sscm_event_app},
\[
(1-\rho)\mathbf{I}_p\preceq \mathbf{H}_t\preceq(1+\rho)\mathbf{I}_p,
\]
so
\[
(1+\rho)^{-1}\mathbf{I}_p\preceq \mathbf{H}_t^{-1}\preceq(1-\rho)^{-1}\mathbf{I}_p.
\]
Since
\[
\mathbf{A}_t=\mathbf{A}_\star^{1/2}\mathbf{H}_t^{-1}\mathbf{A}_\star^{1/2},
\]
\eqref{eq:relative_to_loewner_app} follows. If \(\rho\le\eta/(1+\eta)\), then
\[
(1-\rho)^{-1}\le1+\eta,
\qquad
(1+\rho)^{-1}\ge1-\eta,
\]
which proves the final assertion. \qedhere
\end{proof}

\subsection{Proof of Theorem~\ref{thm:theory_smsscm}}
\label{app:proof_smsscm}

\begin{lemma}[SSCM-metric exact assignment from close centers]
\label{lem:sscm_metric_exact}
Assume \(\mathcal E_n\cap\mathcal M_n\) holds. Suppose \(\bm{m}_1,\ldots,\bm{m}_K\) satisfy
\begin{equation}
\max_{1\le k\le K}\|\bm{m}_k-\bm{\mu}_{\sigma(k)}\|\le a_n
\label{eq:sscm_center_close_app}
\end{equation}
for some permutation \(\sigma\). If \eqref{eq:theory_sep_sscm} holds, then the assignment rule
\[
\bm{X}_i\mapsto \argmin_k (\bm{X}_i-\bm{m}_k)^\top \mathbf{A}_t(\bm{X}_i-\bm{m}_k)
\]
recovers the true partition up to the same permutation.
\end{lemma}

\begin{proof}
Because \(\mathbf{\Gamma}_{\rm sgn}\succeq0\),
\begin{equation}
\mathbf{\Sigma}_\star=\mathbf{\Gamma}_{\rm sgn}+\lambda \mathbf{I}_p\succeq \lambda \mathbf{I}_p,
\qquad
\mathbf{A}_\star=\mathbf{\Sigma}_\star^{-1}\preceq \lambda^{-1}\mathbf{I}_p.
\label{eq:Astar_upper_app}
\end{equation}
Fix \(\bm{X}_i\in C_{r,T}\), set \(k_r=\sigma^{-1}(r)\), and for \(s\ne r\) set \(k_s=\sigma^{-1}(s)\). From \(\mathcal E_n\), \eqref{eq:sscm_center_close_app}, and \eqref{eq:Astar_upper_app},
\begin{align}
\|\bm{X}_i-\bm{m}_{k_r}\|_{\mathbf{A}_\star}
&\le \|\bm{X}_i-\bm{\mu}_r\|_{\mathbf{A}_\star}+\|\bm{m}_{k_r}-\bm{\mu}_r\|_{\mathbf{A}_\star} \notag\\
&\le \lambda^{-1/2}\{\|\bm{X}_i-\bm{\mu}_r\|+\|\bm{m}_{k_r}-\bm{\mu}_r\|\}
\le \frac{2a_n}{\sqrt\lambda}.
\label{eq:sscm_correct_Astar_app}
\end{align}
For \(s\ne r\),
\begin{align*}
\|\bm{X}_i-\bm{m}_{k_s}\|_{\mathbf{A}_\star}
&\ge \|\bm{\mu}_r-\bm{\mu}_s\|_{\mathbf{A}_\star}
     -\|\bm{X}_i-\bm{\mu}_r\|_{\mathbf{A}_\star}
     -\|\bm{m}_{k_s}-\bm{\mu}_s\|_{\mathbf{A}_\star} \notag\\
&\ge d_{\rm sgn}-\frac{2a_n}{\sqrt\lambda}.
\end{align*}
On \(\mathcal M_n\), \eqref{eq:theory_metric_stability} gives, for all \(\bm{v}\),
\begin{equation}
\sqrt{1-\eta}\,\|\bm{v}\|_{\mathbf{A}_\star}
\le \|\bm{v}\|_{\mathbf{A}_t}
\le \sqrt{1+\eta}\,\|\bm{v}\|_{\mathbf{A}_\star}.
\label{eq:sscm_norm_comparison_app}
\end{equation}
Combining \eqref{eq:sscm_correct_Astar_app}--\eqref{eq:sscm_norm_comparison_app},
\begin{equation*}
\|\bm{X}_i-\bm{m}_{k_r}\|_{\mathbf{A}_t}
\le \sqrt{1+\eta}\,\frac{2a_n}{\sqrt\lambda},
\end{equation*}
and for \(s\ne r\),
\begin{equation*}
\|\bm{X}_i-\bm{m}_{k_s}\|_{\mathbf{A}_t}
\ge \sqrt{1-\eta}\left(d_{\rm sgn}-\frac{2a_n}{\sqrt\lambda}\right).
\end{equation*}
Condition \eqref{eq:theory_sep_sscm} is equivalent to
\begin{equation*}
\sqrt{1-\eta}\left(d_{\rm sgn}-\frac{2a_n}{\sqrt\lambda}\right)
>
\sqrt{1+\eta}\,\frac{2a_n}{\sqrt\lambda}.
\end{equation*}
Therefore
\[
(\bm{X}_i-\bm{m}_{k_r})^\top \mathbf{A}_t(\bm{X}_i-\bm{m}_{k_r})
<
(\bm{X}_i-\bm{m}_{k_s})^\top \mathbf{A}_t(\bm{X}_i-\bm{m}_{k_s}),
\qquad s\ne r,
\]
which proves exact assignment. \qedhere
\end{proof}

\begin{proof}[Proof of Theorem~\ref{thm:theory_smsscm}]
Define
\[
\mathcal G_n:=\mathcal E_n\cap\mathcal F_n\cap\mathcal M_n.
\]
By Lemma~\ref{lem:theory_En}, \eqref{eq:Fn_prob_app}, and Assumption~\ref{as:theory_A3},
\begin{equation*}
\Pr(\mathcal G_n^c)
\le \delta_n+K(1-\pi_{\min})^n+\Pr(\mathcal M_n^c)
\to0.
\end{equation*}
Work on \(\mathcal G_n\). Since \(d>4a_n\), Lemma~\ref{lem:seed_hits_all} gives max--min seeds \(\widetilde{\bm{m}}_1,\ldots,\widetilde{\bm{m}}_K\) and a permutation \(\sigma\) such that
\begin{equation*}
\widetilde{\bm{m}}_k\in\{\bm{X}_i:i\in C_{\sigma(k),T}\},
\qquad
\max_k\|\widetilde{\bm{m}}_k-\bm{\mu}_{\sigma(k)}\|\le a_n.
\end{equation*}
Applying Lemma~\ref{lem:sscm_metric_exact} to \(\{\widetilde{\bm{m}}_k\}\) yields
\begin{equation*}
\widehat C_k^{(0)}=C_{\sigma(k),T},
\qquad k=1,\ldots,K.
\end{equation*}
After the spatial-median update, Lemma~\ref{lem:median_stability} gives
\begin{equation*}
\max_k\|\bm{m}_k^{(1)}-\bm{\mu}_{\sigma(k)}\|\le a_n.
\end{equation*}
Applying Lemma~\ref{lem:sscm_metric_exact} again gives
\begin{equation*}
\widehat C_k^{(1)}=C_{\sigma(k),T},
\qquad k=1,\ldots,K.
\end{equation*}
Thus \(\widehat C^{(1)}=\widehat C^{(0)}\). Applying the deterministic spatial-median update to the same partition gives
\begin{equation*}
(\bm{m}_1^{(2)},\ldots,\bm{m}_K^{(2)})=(\bm{m}_1^{(1)},\ldots,\bm{m}_K^{(1)}),
\qquad
\widehat{\mathbf{S}}_2=\widehat{\mathbf{S}}_1,
\qquad
\mathbf{A}_2=\mathbf{A}_1,
\qquad
\widehat C^{(2)}=\widehat C^{(1)}.
\end{equation*}
By induction,
\begin{equation*}
(\widehat C^{(t)},\bm{m}^{(t)},\mathbf{A}_t)=(\widehat C^{(1)},\bm{m}^{(1)},\mathbf{A}_1),
\qquad t\ge1.
\end{equation*}
Finally,
\begin{equation*}
\Pr\left(\exists\sigma:\ \widehat C_k=C_{\sigma(k),T}\ \forall k\right)
\ge \Pr(\mathcal G_n)
\to1,
\end{equation*}
which proves clustering consistency and finite termination with probability tending to one. \qedhere
\end{proof}

\subsection{Proof of Theorem~\ref{thm:theory_feature}}
\label{app:proof_feature}

\begin{lemma}[Uniform perturbation of coordinate separation scores]
\label{lem:score_pert_app}
Fix a coordinate \(j\). If
\[
\max_{1\le k\le K}|m_{kj}-\mu_{kj}|\le b,
\]
then
\begin{equation*}
|s_j(\bm{m}_1,\ldots,\bm{m}_K)-s_j^\star|\le 2Kb.
\end{equation*}
\end{lemma}

\begin{proof}
Let
\[
e_{kj}:=m_{kj}-\mu_{kj},
\qquad
\bar e_j:=K^{-1}\sum_{k=1}^K e_{kj}=\bar m_j-\bar\mu_j.
\]
Then \(|e_{kj}|\le b\) and \(|\bar e_j|\le b\). Therefore
\begin{align*}
|s_j-s_j^\star|
&=\left|\sum_{k=1}^K |m_{kj}-\bar m_j|-
        \sum_{k=1}^K |\mu_{kj}-\bar\mu_j|\right| \notag\\
&\le \sum_{k=1}^K
\left||\mu_{kj}-\bar\mu_j+e_{kj}-\bar e_j|-|\mu_{kj}-\bar\mu_j|\right| \notag\\
&\le \sum_{k=1}^K |e_{kj}-\bar e_j|
\le \sum_{k=1}^K(|e_{kj}|+|\bar e_j|)
\le 2Kb.
\end{align*}
\qedhere
\end{proof}

\begin{proof}[Proof of Theorem~\ref{thm:theory_feature}]
Assume \eqref{eq:theory_center_close} and choose a permutation $\sigma$ satisfying it. Define
\begin{equation*}
\widetilde{\bm{m}}_r:=\bm{m}_{\sigma^{-1}(r)},
\qquad r=1,\ldots,K.
\end{equation*}
For each coordinate $j$,
\begin{equation*}
s_j(\bm{m}_1,\ldots,\bm{m}_K)=s_j(\widetilde{\bm{m}}_1,\ldots,\widetilde{\bm{m}}_K),
\qquad
\max_{1\le r\le K}|\widetilde m_{rj}-\mu_{rj}|\le a_n.
\end{equation*}
Applying Lemma~\ref{lem:score_pert_app} to $\widetilde{\bm{m}}_1,\ldots,\widetilde{\bm{m}}_K$ gives
\begin{equation}
\max_{1\le j\le p}|s_j(\bm{m}_1,\ldots,\bm{m}_K)-s_j^\star|\le 2Ka_n.
\label{eq:uniform_score_pert_app}
\end{equation}
For \(j\notin S_0\), \(s_j^\star=0\), so by \eqref{eq:uniform_score_pert_app} and \eqref{eq:theory_tau_margin},
\begin{equation}
s_j(\bm{m}_1,\ldots,\bm{m}_K)\le 2Ka_n<\tau_n.
\label{eq:noise_excluded_app}
\end{equation}
For \(j\in S_0\), \(s_j^\star\ge s_{\min}\), so
\begin{equation}
s_j(\bm{m}_1,\ldots,\bm{m}_K)\ge s_{\min}-2Ka_n>\tau_n.
\label{eq:signal_included_app}
\end{equation}
Equations \eqref{eq:noise_excluded_app}--\eqref{eq:signal_included_app} give
\[
S(\tau_n)=\{j:s_j(\bm{m}_1,\ldots,\bm{m}_K)\ge\tau_n\}=S_0.
\]
If \eqref{eq:theory_center_close} holds with probability tending to one, then
\[
\Pr\{S(\tau_n)=S_0\}
\ge
\Pr\left\{\max_k\|\bm{m}_k-\bm{\mu}_{\sigma(k)}\|\le a_n\ \text{for some }\sigma\right\}
\to1.
\]
\qedhere
\end{proof}

\subsection{Proof of Theorem~\ref{thm:theory_sparse}}
\label{app:proof_sparse}

Under Assumption~\ref{as:theory_A4}, for every \(r\ne s\),
\begin{equation*}
\|\bm{\mu}_r-\bm{\mu}_s\|^2
=\sum_{j\in S_0}(\mu_{rj}-\mu_{sj})^2
 +\sum_{j\notin S_0}(\mu_{rj}-\mu_{sj})^2
=\|\bm{\mu}_{r,S_0}-\bm{\mu}_{s,S_0}\|^2,
\end{equation*}
and therefore
\begin{equation}
d=d_0.
\label{eq:d_equals_d0_app}
\end{equation}

\begin{proof}[Proof of Theorem~\ref{thm:theory_sparse}]
Let
\[
\mathcal G_n:=\mathcal E_n\cap\mathcal F_n.
\]
By Lemma~\ref{lem:theory_En} and \eqref{eq:Fn_prob_app},
\begin{equation}
\Pr(\mathcal G_n^c)
\le \delta_n+K(1-\pi_{\min})^n
\to0.
\label{eq:Gn_sparse_prob_app}
\end{equation}
Work on \(\mathcal G_n\). From \eqref{eq:d_equals_d0_app} and Assumption~\ref{as:theory_A5},
\[
d=d_0>4a_n.
\]
Lemma~\ref{lem:seed_hits_all} gives seeds \(\widetilde{\bm{m}}_1,\ldots,\widetilde{\bm{m}}_K\) and a permutation \(\sigma\) satisfying
\begin{equation}
\widetilde{\bm{m}}_k\in\{\bm{X}_i:i\in C_{\sigma(k),T}\},
\qquad
\max_k\|\widetilde{\bm{m}}_k-\bm{\mu}_{\sigma(k)}\|\le a_n.
\label{eq:sparse_seed_close_app}
\end{equation}
Theorem~\ref{thm:theory_feature} applied to \eqref{eq:sparse_seed_close_app} gives
\begin{equation}
S^{(0)}=S_0.
\label{eq:sparse_initial_support_app}
\end{equation}
For the first restricted assignment, fix \(\bm{X}_i\in C_{r,T}\) and set \(k_r=\sigma^{-1}(r)\). Since
\[
\|\bm{X}_{i,S_0}-\bm{\mu}_{r,S_0}\|
\le \|\bm{X}_i-\bm{\mu}_r\|
\le a_n,
\qquad
\|\widetilde{\bm{m}}_{k,S_0}-\bm{\mu}_{\sigma(k),S_0}\|
\le a_n,
\]
we have
\begin{align}
\|\bm{X}_{i,S_0}-\widetilde{\bm{m}}_{k_r,S_0}\|&\le 2a_n, \label{eq:sparse_first_correct_app}\\
\|\bm{X}_{i,S_0}-\widetilde{\bm{m}}_{k_s,S_0}\|&\ge d_0-2a_n,
\qquad s\ne r.
\label{eq:sparse_first_wrong_app}
\end{align}
Because \(d_0>4a_n\), \eqref{eq:sparse_first_correct_app}--\eqref{eq:sparse_first_wrong_app} imply
\begin{equation}
\widehat C_k^{(0)}=C_{\sigma(k),T},
\qquad k=1,\ldots,K.
\label{eq:sparse_first_exact_app}
\end{equation}

Assume for some \(t\ge0\) that
\begin{equation}
S^{(t)}=S_0,
\qquad
\widehat C_k^{(t)}=C_{\sigma(k),T},
\quad k=1,\ldots,K.
\label{eq:sparse_induction_hyp_app}
\end{equation}
Then the spatial-median update and Lemma~\ref{lem:median_stability} give
\begin{equation*}
\max_k\|\bm{m}_k^{(t+1)}-\bm{\mu}_{\sigma(k)}\|\le a_n.
\end{equation*}
Theorem~\ref{thm:theory_feature} gives
\begin{equation*}
S^{(t+1)}=S_0.
\end{equation*}
Repeating \eqref{eq:sparse_first_correct_app}--\eqref{eq:sparse_first_wrong_app} with \(\bm{m}_k^{(t+1)}\) in place of \(\widetilde{\bm{m}}_k\) gives
\begin{equation}
\widehat C_k^{(t+1)}=C_{\sigma(k),T},
\qquad k=1,\ldots,K.
\label{eq:sparse_assignment_induction_app}
\end{equation}
Thus \eqref{eq:sparse_induction_hyp_app} propagates by induction, and \eqref{eq:sparse_initial_support_app}--\eqref{eq:sparse_first_exact_app} give the base case.

After \eqref{eq:sparse_assignment_induction_app} with $t=0$, the partition and support are unchanged from the first exact restricted assignment:
\[
S^{(1)}=S^{(0)}=S_0,
\qquad
\widehat C_k^{(1)}=\widehat C_k^{(0)}=C_{\sigma(k),T}.
\]
The next deterministic spatial-median update is therefore computed from the same sets, and hence
\begin{equation*}
(\bm{m}_1^{(2)},\ldots,\bm{m}_K^{(2)})=(\bm{m}_1^{(1)},\ldots,\bm{m}_K^{(1)}),
\qquad
S^{(2)}=S^{(1)},
\qquad
\widehat C^{(2)}=\widehat C^{(1)}.
\end{equation*}
By induction, the same state is reproduced thereafter. Therefore, on \(\mathcal G_n\), the procedure reaches a fixed point with active set \(S_0\) and the true partition up to permutation. With \eqref{eq:Gn_sparse_prob_app},
\[
\Pr\left(S(\tau_n)=S_0,
\ \exists\sigma:\widehat C_k=C_{\sigma(k),T}\ \forall k\right)
\ge \Pr(\mathcal G_n)
\to1.
\]
\qedhere
\end{proof}

\subsection{Verification of the explicit sufficient conditions}
\label{app:explicit_conditions}

Because \(\delta_n=n^{-\varepsilon}\), under (T$_\alpha$)
\begin{equation}
a_n
=\sqrt{\lambda_{\max}}\,\kappa\{\log(2n/\delta_n)\}^{1/\alpha}
=\sqrt{\lambda_{\max}}\,\kappa\{\log(2n^{1+\varepsilon})\}^{1/\alpha}
=:B_{\alpha,n}.
\label{eq:Balpha_equals_an_app}
\end{equation}
Under (T$_\nu$),
\begin{equation}
a_n
=\sqrt{\lambda_{\max}}(Cn/\delta_n)^{1/\nu}
=\sqrt{\lambda_{\max}}(Cn^{1+\varepsilon})^{1/\nu}
=:B_{\nu,n}.
\label{eq:Bnu_equals_an_app}
\end{equation}
Substituting \eqref{eq:Balpha_equals_an_app} or \eqref{eq:Bnu_equals_an_app} into Assumptions~\ref{as:theory_A3} and~\ref{as:theory_A5} gives the displayed sufficient conditions in Section~\ref{subsec:theory_explicit}. In particular,
\begin{equation*}
s_{\min}>4Ka_n
\quad\Longrightarrow\quad
(2Ka_n,\ s_{\min}-2Ka_n)\ne\varnothing,
\end{equation*}
so a deterministic threshold sequence satisfying \eqref{eq:theory_tau_margin} exists under the stated minimum-signal condition. The SSCM part is the metric-stability requirement \eqref{eq:theory_metric_stability}, which is implied by the relative SSCM bound in Lemma~\ref{lem:sscm_spectral_app}.

\section{Additional Simulation Results}
\label{app:smsscm_sim}

We also consider an additional weakly sparse setting with strong feature dependence and heterogeneous scales.
Specifically, we generate data from a three-component mixture model with \(K=3\) clusters and \(n_0=100\) observations per cluster. The ambient dimension is taken as \(p\in\{100,200\}\). The cluster means are given by
\[
\bm{\mu}_1=\bm{0}, \qquad
\bm{\mu}_2=(\delta \bm{1}_{s_p},0,\ldots,0), \qquad
\bm{\mu}_3=(-\delta \bm{1}_{s_p},0,\ldots,0),
\]
where \(s_p=p/4\), so that the signal is only weakly sparse compared with the main simulation settings. We set \(\delta=1.2\).

The covariance matrix is constructed as
\[
\mathbf{\Sigma} = \mathbf{D}^{1/2} \mathbf{R} \mathbf{D}^{1/2},
\]
where
\[
(\mathbf{R})_{ij}=\rho^{|i-j|}, \qquad \rho=0.85,
\]
and
\[
\mathbf{D}=\mathrm{diag}(d_1,\ldots,d_p), \qquad
d_j=
\begin{cases}
1, & 1\le j\le p/3,\\
4, & p/3 < j\le 2p/3,\\
9, & 2p/3 < j\le p.
\end{cases}
\]
This choice induces both strong feature dependence and substantial scale heterogeneity, under which Euclidean assignment is no longer well aligned with the underlying cluster geometry.

We consider two distributional settings. In the first, observations are generated from a Gaussian mixture,
\[
\bm{X}_i \sim \mathcal{N}(\bm{\mu}_k,\mathbf{\Sigma}), \qquad i\in C_k.
\]
In the second, we replace the Gaussian distribution by a multivariate \(t\)-distribution with \(\nu=5\) degrees of freedom to introduce moderate heavy-tailed behavior.

To evaluate clustering performance, we consider five standard external measures, summarized in Table~\ref{tab:clustering_metrics}.
The results are reported in Table~\ref{tab:metrics_two_dist}, with all values averaged over 100 replications.
Table~\ref{tab:metrics_two_dist} shows that the two proposed methods consistently outperform the classical baselines under both Gaussian and $t_5$ settings.
In particular, Sparse--SM attains the best performance across all five evaluation measures, while SM--SSCM remains clearly more competitive than K-means, K-median, and sparse K-means. Compared with the main simulations, these results suggest that SM--SSCM becomes more effective when strong dependence and heterogeneous scales are present, whereas Sparse--SM continues to benefit from explicit feature exclusion even when the signal is only weakly sparse.

\begin{table}[htbp]
\centering
\caption{External clustering evaluation measures used in the simulation study. For all measures, larger values indicate better clustering performance.}
\label{tab:clustering_metrics}
\renewcommand{\arraystretch}{1.15}
\setlength{\tabcolsep}{5pt}
\small
\begin{tabular}{p{2.2cm} p{9.2cm} p{2cm}}
\toprule
\textbf{Metric} & \textbf{Definition} & \textbf{Range} \\
\midrule
ARI
& Chance-adjusted version of the Rand index, measuring pairwise agreement between the estimated and true partitions.
& $[-1,1]$
\\

Purity
& $\displaystyle \text{Purity}=\frac{1}{n}\sum_k \max_j |C_k\cap L_j|$, where $\{C_k\}$ and $\{L_j\}$ denote the estimated clusters and true classes, respectively.
& $[0,1]$
 \\

NMI
& $\displaystyle \text{NMI}(C,L)=\frac{2I(C;L)}{H(C)+H(L)}$, where $I(C;L)$ is the mutual information between the estimated and true labels, and $H(\cdot)$ denotes Shannon entropy.
& $[0,1]$
 \\

FMI
& $\displaystyle \text{FMI}=\frac{\mathrm{TP}}{\sqrt{(\mathrm{TP}+\mathrm{FP})(\mathrm{TP}+\mathrm{FN})}}$, where TP, FP, and FN denote the numbers of true-positive, false-positive, and false-negative pairs, respectively.
& $[0,1]$
 \\

V-measure
& Harmonic mean of homogeneity and completeness.
& $[0,1]$
\\
\bottomrule
\end{tabular}
\end{table}

\begin{table}[htbp]
\centering
\caption{Clustering performance under Gaussian and $t_5$ distributions. Reported values are averages over repeated runs.}
\label{tab:metrics_two_dist}
\setlength{\tabcolsep}{2.8pt}
\renewcommand{\arraystretch}{1.1}
\footnotesize
\begin{tabular}{clcccccc}
\toprule
\textbf{Dist.} & \textbf{Metric}
& \textbf{SM--SSCM} & \textbf{Sparse--SM} & \textbf{K-means} & \textbf{K-median} & \textbf{SKMeans} & \textbf{SKMedian} \\
\midrule
\multirow{5}{*}{Gaussian}
& ARI & 0.399(0.034)&\textbf{0.471(0.083)}&  0.014(0.021)& 0.025(0.030)& -0.0003(0.004)&  0.295(0.118)\\
& Purity& 0.696(0.036)& \textbf{0.754(0.062)}&0.402(0.034)&0.417(0.042)&0.371(0.015)&0.622(0.099)\\
& NMI&0.417(0.028)&\textbf{0.471(0.064)}&0.020(0.021)& 0.031(0.031) &0.005(0.004)&0.314(0.124)\\
& FMI &0.601(0.022)&\textbf{0.649(0.053)}&0.344(0.015)&0.352(0.020)& 0.336(0.005)&0.531(0.079)\\
& V-measure &0.420(0.027)&\textbf{0.474(0.062)}&0.020(0.021)&0.031(0.031)& 0.005(0.004)&0.316(0.125)\\
\midrule
\multirow{5}{*}{$t_5$}
& ARI& 0.349(0.031)&\textbf{0.374(0.045)}&0.005(0.009)&0.014(0.030)&-0.0003(0.004)&0.262(0.098)\\
& Purity&0.667(0.030)&\textbf{0.692(0.043)}& 0.384(0.021)&0.394(0.042)&0.367(0.014)&0.603(0.081)\\
& NMI&0.350(0.033)&\textbf{0.372(0.039)}&0.011(0.008)&0.019(0.030)&0.005(0.003)&0.272(0.100)\\
& FMI&0.569(0.020)&\textbf{0.585(0.029)}&0.366(0.034) &0.377(0.059)&0.376(0.029)&0.508(0.065) \\
& V-Measure&0.354(0.033)&\textbf{0.375(0.038)}&0.012(0.008)&0.020(0.030)&  0.006(0.003)&0.273(0.101)\\
\bottomrule
\end{tabular}
\end{table}


\clearpage

\bibliographystyle{asa}
\bibliography{ref}

\end{document}